\documentclass[journal]{IEEEtran}  

\IEEEoverridecommandlockouts                              

\usepackage{cite}
\usepackage{amssymb,amsmath}
\usepackage[dvipdfmx]{graphicx}
\usepackage{color}
\usepackage{url}

\newtheorem{secthm}{Theorem}[section]
\newtheorem{seclem}[secthm]{Lemma}
\newtheorem{seccor}[secthm]{Corollary}
\newtheorem{secdefn}[secthm]{Definition}
\newtheorem{secrem}[secthm]{Remark}
\newtheorem{secprop}[secthm]{Proposition}
\newtheorem{secass}[secthm]{Assumption}
\def\red{\hfill $\lhd$}

\newcommand{\bR} { {\mathbb R}}

\newcommand{\bP} { {\mathbb P}}
\newcommand{\bZ} { {\mathbb Z}}

\newcommand{\cM} { {\mathcal M}}
\newcommand{\cN} { {\mathcal N}}
\newcommand{\cO} { {\mathcal O}}
\newcommand{\cQ} { {\mathsf Q}}
\newcommand{\cR} { {\mathsf R}}
\newcommand{\cS} { {\mathcal S}}

\newcommand{\cK} { {\mathcal K}}
\newcommand{\cF} { {\mathcal F}}
\newcommand{\cU} { {\mathcal U}}
\newcommand{\cX} { {\mathcal X}}

\title{
Design of Privacy-Preserving Dynamic Controllers
}

\author{Yu Kawano and Ming Cao
\thanks{This work was supported in part by the European Research Council (ERC-CoG-771687) and the Dutch Organization for Scientific Research (NWO-vidi-14134).}
\thanks{Y. Kawano is with the Graduate School of Advanced Science and Engineering, Hiroshima University, Higashi-Hiroshima, Japan (email: ykawano@hiroshima-u.ac.jp).}
\thanks{Ming Cao is with the Faculty of Science and Engineering, University of Groningen, 9747 AG Groningen, The Netherlands  (email: m.cao@rug.nl).}%
}

\begin{document}
\allowdisplaybreaks[4]

\maketitle
\thispagestyle{empty}
\pagestyle{empty}

\begin{abstract}
As a quantitative criterion for privacy of ``mechanisms'' in the form of data-generating processes, the concept of \emph{differential privacy} was first proposed in computer science and has later been applied to linear dynamical systems. However, differential privacy has not been studied in depth together with other properties of dynamical systems, and it has not been fully utilized for controller design. In this paper, first we clarify that a classical concept in systems and control, \emph{input observability} (sometimes referred to as \emph{left invertibility}) has a strong connection with differential privacy. In particular, we show that the Gaussian mechanism can be made highly differentially private by adding \emph{small} noise if the corresponding system is less input observable. Next, enabled by our new insight into privacy, we develop a method to design dynamic controllers for the classic  tracking control problem while addressing privacy concerns. We call the obtained controller through our design method the \emph{privacy-preserving controller}. The usage of such controllers is further illustrated by an example of tracking the prescribed power supply in a DC microgrid installed with smart meters while keeping the electricity consumers' tracking errors private.
\end{abstract}

\section{Introduction}
The trend of the Internet-of-Things (IoT) and cloud computing makes privacy and security become a research area of acute social and technological concerns, see e.g.~\cite{AIM:10,Weber:10,GBM:13,WWR:10,Ryan:11,TJA:10,Mervis:19}. To protect the privacy of data sources, the collected data are usually processed statistically before being publicized for different applications. However, even if one only publishes statistical analytics, not raw data, private personal information may still be identified by smart data mining algorithms that combine the statistics with other third party information, see e.g.~\cite{Sweeney:97,MS:04,Hansell:06,NS:06}. Motivated by threats on privacy, statistical disclosure control, or more generally privacy preserving data mining, has been intensively studied; see e.g.~\cite{WD:96,WD:12}. Representative techniques include the $K$-anonymity~\cite{Sweeney:02}, $l$-diversity~\cite{MGK:06}, $t$-closeness~\cite{LLV:07}, and differential privacy~\cite{DMN:06,DKM:06}. In particular, differential privacy enjoys the mathematical property of being quantifiable and thus has been used in solving various privacy-related problems arising in the domains of smart grids~\cite{ZJW:14,JZC:13,SDT:15}, health monitoring~\cite{DE:12,DFE:13}, blockchain (or bitcoin)~\cite{YMH:18,AKR:13} and mechanism design~\cite{MT:07}.

There is a growing need to treat privacy as a critical property of \emph{dynamical} systems instead of the feature of some \emph{static} time invariant data set. For example, in power grids, consumers' electricity consumption patterns change over time and are coupled in a closed loop with the stabilization actions of various controllers in  power systems. To address privacy issues of those datasets that are generated by dynamical systems, the standard concept of \emph{differential privacy} for static data has been extended to discrete-time linear dynamical systems, see e.g.~\cite{NP:14,LM:18}, which shows convincingly that the key idea of differential privacy, namely adding noise to data before publishing them, is also effective for privacy protection for dynamical data sets. However, there is still a considerable lack of in-depth understanding of the possible fundamental  interplay between differential privacy and other critical properties of dynamical systems~\cite{KOV:17}.

To address this challenge, we propose to take an approach that is deeply rooted in systems and control theory; to be more specific, we study privacy of dynamical systems by taking two major steps: first to study privacy in terms of \emph{input observability} and then to provide a privacy-preserving controller design method. The differential privacy level of a discrete-time linear system can be interpreted as a quantitative criterion for the difficulty of identifying its input, which triggers us to give a refreshing look at rich classic results on uniquely determining the input from the output in systems and control under the name of input observability~\cite{HP:98} or left invertibility~\cite{SM:69}. For input observability, there are already several qualitative criteria, e.g. the rank condition of the transfer function matrix~\cite{SM:69}, the PBH type test~\cite{Moylan:77,HP:98}, and Kalman's rank type conditions~\cite{MS:68,SM:69}. However, these existing conditions do not provide quantitative analysis. Therefore, there is a gap between the relatively new concept of differential privacy and the classical concept of input observability. To establish a bridge between this gap, we extend the notion of the Gramian to input observability. Then, we show that the Gaussian mechanism evaluates the maximum eigenvalue of the input observability Gramian; in other words, \emph{small} noise is enough to make the less input observable Gaussian mechanism highly differentially private. This new insight suggests that the input observability Gramian can be used for detailed privacy analysis, not restricted to differential privacy, just like what the standard controllability and observability Gramians can do for detailed controllability and observability analysis.

Next, we consider achieving trajectory tracking while protecting the tracking error as private information. Trajectory tracking itself has been studied as a part of the output regulation problem~\cite{Huang:04} for which dynamic output feedback controllers have been studied. The differential privacy level increases if the dynamic controllers are designed such that the maximum eigenvalue of the input observability Gramian is small, which is achieved by making the corresponding $H_{\infty}$-norm small.  In this paper, we provide a dynamic controller design method in order to address the tracking problem and to specify the $H_{\infty}$-norm simultaneously based on LMIs. It is worth pointing out that to increase the differential privacy level of the controller, one needs to make the $H_{\infty}$-norm of the controller small or add large noise, both of which may deteriorate the control performance. Therefore, privacy-preserving controller design reduces to a trade-off between the privacy level and control performance. 

Along this line of research on designing privacy-preserving controllers, there are related earlier works.  Differential privacy has been employed for privacy-preserving filtering~\cite{NP:14,LM:18}, but not for controller design. In particular,~\cite{NP:14} also studies the connection between differential privacy and the $H_{\infty}$-norm of a system; however, differential privacy has not been studied from the input observability perspective, which was considered in our preliminary conference version~\cite{KC:18}. Different from~\cite{NP:14,KC:18}, in this paper we consider not just i.i.d. noise; although this may seem to be a rather minor technical extension, it is in fact an important step towards obtaining a deeper understanding of the differential privacy level of a dynamical system. Also note that differential privacy has been used for LQ control~\cite{YJ:18} and distributed optimization~\cite{HE:17,HMD:12,HMV:15,HT:17,SCS:13}, where the controller gains or controller dynamics are designed without considering privacy issues, and consequently privacy-preserving noise is added separately, making protecting privacy independent of the controller design itself. In contrast, we design the controller with the incorporated goal of achieving high privacy levels using small noise.

The remainder of this paper is organized as follows. Section~\ref{DPA:s} introduces the concept of differential privacy and analyzes it from several aspects including input observability. Section~\ref{PCD:s} provides a privacy-preserving controller design method. Our method is illustrated by an example of DC microgrids installed with smart meters in Section~\ref{Ex:s}. Section~\ref{NDP:s} briefly mentions extensions of our results to nonlinear systems, where a part of the results has been presented in a preliminary conference version~\cite{KC2:18}. Finally, Section~\ref{Con:s} concludes the paper.

{\bf Notations:} The set of real numbers, non-negative real numbers, and non-negative integers are denoted by $\bR$, $\bR_+$ and $\bZ_+$, respectively. For vectors $x_1,\dots,x_m \in \bR^n$, a collective vector $[x_1^\top \ \cdots \ x_m^\top]^\top \in \bR^{nm}$ is also described by $[x_1 ; \cdots ; x_m]$ for the sake of simplicity of description. For the sequence $u(t) \in \bR^m$, $t \in \bZ_+$, a collective vector consisting of its subsequence is denoted by $U_t(\tau):=[u(\tau) ; \cdots ; u(\tau +t)] \in \bR^{(t+1)m}$; when $\tau =0$, the argument is omitted, i.e., $U_t:=[u(0) ; \cdots ; u(t)]$. For a square matrix $A \in \bR^{n \times n}$, its determinant is denoted by ${\rm det}(A)$, and when its eigenvalues are real, its maximum and minimum eigenvalues are denoted by $\lambda_{\max}(A)$ and $\lambda_{\min}(A)$, respectively. Further,~$A\succ 0$ means that~$A$ is symmetric and positive definite. The identity matrix of size $n$ is denoted by $I_n$. For the vector $x\in\bR^n$, its norms is denoted by $| x |_p:= \left( \sum_{i=1}^n| x_i|^p \right)^{1/p}$, where~$p \in \bZ_+$, and its weighted norm with $A \succ 0$ is denoted by $|x|_A:= (x^\top A x)^{1/2}$.  A continuous function $\alpha:[0,a) \to \bR_+$ is said to be of class $\cK$ if it is strictly increasing and $\alpha(0)=0$.  Moreover, it is said to be of class $\cK_\infty$ if $a = \infty$ and $\alpha (r) \to \infty$ as $r \to \infty$. A random variable $w$ is said to have a non-degenerate multivariate Gaussian distribution with the mean value $\mu \in \bR^n$ and covariance matrix $\Sigma \succ 0$, denoted by $w \sim \cN_n (\mu,\Sigma)$, if its distribution has the following probability density:
\begin{align*}
p(w;\mu,\Sigma) = \left(\frac{1}{(2 \pi)^n {\rm det}(\Sigma)}\right)^{1/2} {\rm e}^{- |w - \mu|_{ \Sigma^{-1}}^2/2}.
\end{align*}
The so called $\cQ$-function is defined by $\cQ (w) := \frac{1}{\sqrt{2\pi}}\int_w^{\infty} {\rm e}^{-\frac{v^2}{2}}dv$, where $\cQ(w) <1/2$ for $w>0$, and $\cR (\varepsilon,\delta):=(\cQ^{-1}(\delta) + \sqrt{(\cQ^{-1}(\delta))^2 + 2 \varepsilon})/2\varepsilon$.

\section{Differential Privacy Analysis}\label{DPA:s}
In this section, we study differential privacy of discrete-time linear dynamical systems from three aspects. First, we define the differential privacy of a Gaussian mechanism with output noise~\cite{DMN:06,DKM:06}; the exact definition of a mechanism will become clear later. Second, we investigate the differential privacy of the mechanism in terms of observability. Last, we analyze the differential privacy of the mechanism with input noise. Throughout the paper, we follow the convention by focusing on a finite data sets. In a dynamical system setting, this corresponds to analyzing the system's properties within a finite time.

Consider the following discrete-time linear system:
\begin{align}
\left\{\begin{array}{l}
x(t+1) = A x(t) + B u(t), \\
y(t) = C x(t) + D u(t), 
\end{array}\right. \label{sys}
\end{align}
for $t \in \bZ_+$, where $x(t)\in\bR^n$, $u(t) \in \bR^m$ and $y(t) \in \bR^q$ denote the state, input and output, respectively, and $A \in \bR^{n \times n}$, $B \in \bR^{n \times m}$, $C \in \bR^{q \times n}$ and $D \in \bR^{q \times m}$. 

For (\ref{sys}), the output sequence $Y_t \in \bR^{(t+1)q}$ is described by
\begin{align}
Y_t = O_t x_0 + N_t U_t, \label{Outseq}
\end{align}
where $O_t \in \bR^{(t+1)q \times n}$ and $N_t \in \bR^{(t+1)q \times (t+1)m}$ are
\begin{align}
&O_t:=\left[\begin{array}{cccc}
C^\top & CA^\top & \cdots & (CA^t)^\top
\end{array}\right]^\top, \label{Obmatrix}\\ 
&N_t:=\left[\begin{array}{cccccc}
D & 0 & \cdots & \cdots &  0\\
CB & D & \ddots &  & \vdots \\
CAB & CB & D &\ddots & \vdots \\
\vdots & \vdots & \ddots & \ddots &0 \\
CA^{t-1}B & CA^{t-2}B & \cdots & CB & D
\end{array}\right]. \label{Markov}
\end{align}
To facilitate future discussion, we also denote the left $(t+1)q$ by $(T+1)m$ submatrix of $N_t$ by $N_{t,T}$, $T\le t$.

\begin{secrem}
If~$[\begin{matrix}O_t & N_t\end{matrix}]=0$, then~$Y_t$ is identically zero. In this pathological case, there is no reason to proceed with privacy analysis, and thus throughout the paper we assume that~$[\begin{matrix}O_t & N_t\end{matrix}]\neq 0$. 
\red
\end{secrem}

\subsection{Differential Privacy With Output Noise}~\label{IO:s}
To proceed with differential privacy analysis, we consider the output~$y_w(t):= y(t) + w(t)$ after adding the noise $w(t) \in \bR^q$.
From (\ref{Outseq}), $Y_{w,t} \in \bR^{(t+1)q}$ can be described by
\begin{align}
Y_{w,t} = O_t x_0 + N_t U_t + W_t. \label{Out_noise}
\end{align}
This defines a mapping $\cM :\bR^n \times \bR^{(t+1)m} \times \bR^{(t+1)q} \ni (x_0,U_t,W_t) \mapsto Y_{w,t} \in \bR^{(t+1)q}$. In differential privacy analysis, this mapping is called a \emph{mechanism}~\cite{DMN:06,DKM:06}.  

It is worth clarifying that the input of the dynamical system~(\ref{sys}) is~$u$ while the input data of the induced mechanism~(\ref{Out_noise}) is $(x_0,U_t)$. 
\begin{secrem}
Depending on specific applications, $x_0$ and $U_t$ do not need to be private at the same time. Our results can be readily extended to the scenario where one of $x_0$ and $U_t$ is confidential, and the other is public.
\red
\end{secrem}

Differential privacy gives an index of the privacy level of a mechanism, which is characterized by the sensitivity of the published output data $Y_{w,t}$ with respect to the input data $(x_0,U_t)$. More specifically, if for a pair of not so distinct input data $((x_0,U_t),(x'_0,U'_t))$, the corresponding pair of output data $(Y_{w,t},Y'_{w,t})$ are very different, then one can conclude that input data are easy to identify, i.e. the mechanism is less private. Thus, differential privacy is defined using  a pair of different but ``similar'' input data, where by similar we mean that the pair satisfies the following adjacency relations.  
\begin{secdefn}\label{SBRU:def}
Given~$c>0$ and~$p \in \bZ_+$, a pair of input data~$((x_0,U_t),(x'_0,U'_t)) \allowbreak \in (\bR^n \times \bR^{(t+1)m}) \times (\bR^n \times\bR^{(t+1)m})$ is said to belong to the binary relation \emph{$c$-adjacency} under the $p$ norm if~$| [x_0; U_t] - [x'_0; U'_t]  |_p  \le c$. The set of all pairs of the input data that are $c$-adjacent under the $p$ norm is denoted by~${\rm Adj}_p^c$.
\red
\end{secdefn}

The magnitude of~$c$ gives an upper bound on the difference of the pair of input data $(x_0,U_t)$ and $(x'_0,U'_t)$. Therefore, $c$ can be chosen according to the knowledge of the range or distribution of input data.

Now, we are ready to define differential privacy of the mechanism~\eqref{Out_noise}.
\begin{secdefn}\label{Diffprv:def}
Let ~$(\bR^{(t+1)q},\cF,\bP )$ be a probability space. The mechanism~\eqref{Out_noise} is said to be \emph{$(\varepsilon,\delta)$-differentially private} for ${\rm Adj}_p^c$ at a finite time instant~$t\in \bZ_+$ if there exist $\varepsilon>0$ and $\delta \ge 0$ such that 
\begin{align}
&\bP (O_t x_0 + N_t U_t + W_t \in \cS) \nonumber\\
&\le {\rm e}^{\varepsilon} \bP (O_t x'_0 + N_t U'_t + W_t \in \cS) + \delta, \ \forall \cS \in \cF \label{diffprv}
\end{align} 
for any $((x_0,U_t),(x'_0,U'_t)) \in {\rm Adj}_p^c$.
\red
\end{secdefn}

\begin{secrem}
There are two minor differences between Definition~\ref{SBRU:def} and the symmetric binary relation in~\cite{NP:14}. In~\cite{NP:14}, it is assumed that~$x_0 = x'_0$ in the binary relation and the pair of input sequences~$(U_t,U'_t)$ are the same except for one element in the sequence, which is a special case of Definition~\ref{SBRU:def}.
Our definition of differential privacy is a direct extension of the original one~\cite{DMN:06,DKM:06} and slightly different from that defined for linear dynamical systems in~\cite{NP:14}; our definition depends on the initial state in addition to the input sequence, and $W_t$ is not necessarily causal. 
\red
\end{secrem}

If $\varepsilon$ and $\delta$ are large, then for a different pair of input data $((x_0,U_t),(x'_0, U'_t))$, the corresponding probability distributions of output data $(Y_{w,t},Y'_{w,t})$ can be very different, i.e., a mechanism is less private. Therefore, the privacy level of a mechanism can be evaluated by the pair of variables $\varepsilon$ and $\delta$. From its definition, one notices that if a mechanism is $(\varepsilon_1,\delta_1)$-differentially private, then it is $(\varepsilon_2,\delta_2)$-differentially private for any $\varepsilon_2 \ge \varepsilon_1$ and $\delta_2 \ge \delta_1$. Therefore, $\varepsilon$ and $\delta$ give a lower bound on the privacy level, where larger $\varepsilon$ and $\delta$ imply lower privacy levels. 

As is clear from the definition, $\varepsilon$ and $\delta$ also depend on noise. In fact, we will show that the sensitivity of the dynamical system~(\ref{sys}) provides the lower bound on the covariance matrix for the multivariate Gaussian noise to achieve $(\varepsilon,\delta)$-differential privacy, which is a generalization of~\cite[Theorem~3]{NP:14,DKM:06}. In what follows, we call a mechanism with the Gaussian noise a \emph{Gaussian mechanism}.
\begin{secthm}\label{dpr:thm}
The Gaussian mechanism (\ref{Out_noise}) induced by $W_t \sim \cN_{(t+1)q}(\mu,\Sigma)$ is $(\varepsilon,\delta)$-differentially private for ${\rm Adj}_2^c$ at a finite time $t \in \bZ_+$ with $\varepsilon > 0$ and $1/2 > \delta > 0$ if the covariance matrix $\Sigma \succ 0$ is chosen such that
\begin{align}
\lambda_{\max}^{-1/2}\left(\cO_{\Sigma,t}\right) \ge c \cR (\varepsilon,\delta),\label{ep}
\end{align}
where
\begin{align}
\cO_{\Sigma,t}:=\left[\begin{array}{cc}
O_t & N_t
\end{array}\right]^\top \Sigma^{-1}\left[\begin{array}{cc}
O_t & N_t
\end{array}\right]. \label{inpGtt}
\end{align}
\end{secthm}
\begin{proof}
Using a similar argument as in the proof for~\cite[Theorem~3]{NP:14}, for arbitrary $\varepsilon >0$, one has 
\begin{align*}
&\bP (O_t x_0 +  N_t U_t + W_t \in \cS)\\
&\le {\rm e}^{\varepsilon} \bP (O_t x'_0 +  N_t U'_t + W_t \in \cS) + \bP \left(\tilde W \ge \varepsilon z - 1/2z \right),\end{align*}
where
\begin{align*}
z:=|O_t (x'_0 - x_0) + N_t (U'_t - U_t) |_{ \Sigma^{-1}}^{-1},
\end{align*}
and $\tilde W \sim \cN (0,1)$. Then, the mechanism is $(\varepsilon,\delta)$-differentially private if $\cQ\left(\varepsilon z - \frac{1}{2z}\right) \le \delta$, i.e. 
\begin{align}
z \ge \cR (\varepsilon,\delta), \label{z,delta}
\end{align}
for any $((x_0,U_t),(x'_0,U'_t)) \in {\rm Adj}_2^c$. The inequality (\ref{z,delta}) holds if (\ref{ep}) is satisfied because
\begin{align*}
z^{-1} =& |O_t (x'_0 - x_0) + N_t (U'_t - U_t) |_{ \Sigma^{-1}} \le c \lambda_{\max}^{1/2}\left(\cO_{\Sigma,t}\right).
\end{align*}
\end{proof}

In~\eqref{ep}, only the matrix~$\left[\begin{array}{cc}O_t & N_t\end{array}\right]$ depends on the system dynamics~(\ref{sys}). We will analyze this matrix in terms of system~(\ref{sys})'s input observability in the next subsection. When the initial state (resp. input sequence) is public, the condition~\eqref{ep} can be replaced by $\lambda_{\max}^{-1/2}(N_t^\top \Sigma^{-1} N_t ) \ge c \cR (\varepsilon,\delta)$ (resp. $\lambda_{\max}^{-1/2} ( O_t^\top \Sigma^{-1} O_t ) \ge c \cR (\varepsilon,\delta)$). The matrix~$\cO_{\Sigma,t}$ defined in~\eqref{inpGtt} is in fact the Fisher information matrix of~$Y_t$ with respect to~$[x_0 ; \; U_t]$. Therefore, Theorem~\ref{dpr:thm} connects differential privacy with Fisher information.

From~\eqref{inpGtt},~$\lambda_{\max}^{1/2}(\cO_{\Sigma,t})$ is the~$2$-induced matrix norm of~$\Sigma^{-1/2}[O_t \; N_t]$, denoted by~$|\Sigma^{-1/2}[O_t \; N_t]|_2$. This can be upper bounded as follows.
\begin{align*}
\lambda_{\max}^{1/2}(\cO_{\Sigma,t})&=\left| \Sigma^{-1/2} \begin{bmatrix} O_t & N_t \end{bmatrix} \right|_2\\
& \le \left| \Sigma^{-1/2} \right|_2 \left| \begin{bmatrix}O_t & N_t \end{bmatrix} \right|_2\\
&= \lambda_{\min}^{-1/2} (\Sigma)  \lambda_{\max}^{1/2}\left(\cO_{I_{(t+1)q},t}\right),
\end{align*}
and consequently,
\begin{align}
\lambda_{\max}^{-1/2}(\cO_{\Sigma,t}) \ge \lambda_{\min}^{1/2} (\Sigma)  \lambda_{\max}^{-1/2}\left(\cO_{I_{(t+1)q},t}\right).
\label{ep_sufficient1}
\end{align}
Therefore, for any given $c$, $\varepsilon>0$ and $1/2 > \delta > 0$, one can make the Gaussian mechanism $(\varepsilon,\delta)$-differentially private if one makes the minimum eigenvalue of the covariance matrix~$\Sigma$ sufficiently large such that
\begin{align}
\lambda_{\min}^{1/2} (\Sigma) \ge c \lambda_{\max}^{1/2}\left(\cO_{I_{(t+1)q},t}\right) \cR (\varepsilon,\delta)
\label{ep_sufficient2}
\end{align}
because~\eqref{ep_sufficient1} and~\eqref{ep_sufficient2} imply~(\ref{ep}). In the special case where $\Sigma =\sigma^2 I_{(t+1)q}$, $\sigma>0$ (an i.i.d. Gaussian noise), \eqref{ep_sufficient2} becomes
\begin{align}
\sigma \ge c \lambda_{\max}^{1/2}\left(\cO_{I_{(t+1)q},t}\right) \cR (\varepsilon,\delta) .
\label{sigma}
\end{align}
Still one can design~$\sigma$ to make the Gaussian mechanism $(\varepsilon,\delta)$-differentially private for arbitrary $\varepsilon>0$ and $1/2>\delta>0$.

\begin{secrem}\label{Laplace:rem}
One can also extend~\cite[Theorem~2]{NP:14} to use the i.i.d. Laplace noise in our problem setting. However, the extension to the multivariate Laplace noise is not easy because this involves the computation of the modified Bessel function of the second kind. Let $w_i(t)$, $i=1,\dots,q$, $t \in \bZ_+$ be an i.i.d. Laplace noise with the variance $\mu \in \bR$ and distribution $b >0$.  Then, the Laplace mechanism (\ref{Out_noise}) is $(\varepsilon,0)$-differentially private at a finite time $t$ with $\varepsilon > 0$ if
\begin{align*}
b \ge  c \left| \left[\begin{array}{cc}
O_t & N_t
\end{array}\right]\right|_1 /\varepsilon,
\end{align*}
for any $((x_0,U_t),(x'_0,U'_t)) \in {\rm Adj}_1^c$, where $|A|_1:=\max_j \sum_i |a_{i,j}|$ is the induced matrix $1$-norm. As for the Gaussian mechanism, the induced matrix norm of~$[O_t \; N_t]$ plays a crucial role for the Laplace mechanism too. In the next subsection, we study its~$2$-norm in terms of system~(\ref{sys})'s input observability. Because of the equivalence of induced matrix norms, the observation for the~$2$-norm is applicable to an arbitrary norm including the~$1$-norm.
\red
\end{secrem}

\begin{secrem}
In this subsection, to make the input data private, noise is added to the output data, which makes the output data also private. To analyze the differential privacy level of the output data, one can employ the conventional results for a static data set in~\cite{DMN:06,DKM:06}. By adding a sufficiently large noise, it is possible to achieve the differential privacy requirements for the input data and output data at the same time. \red
\end{secrem}

Note that in Theorem~\ref{dpr:thm}, the system~(\ref{sys}) is not necessarily stable. Now, we focus on asymptotically stable systems. Then, one can characterize the differentially privacy level in terms of the $H_{\infty}$-norm and the observability Gramian, where the $H_{\infty}$-norm of the system~(\ref{sys}) is the infimum non-negative constant $\gamma$ satisfying
\begin{align*}
\sum_{\tau =0}^t |y(\tau )|_2^2 \le \gamma^2 \sum_{\tau =0}^t |u(\tau)|_2^2, \ \forall t \in \bZ_+, 
\end{align*}
for all $L_2$-bounded input signals, and the observability Gramian is
\begin{align}
 \cO_{\infty}:= O_{\infty}^\top O_{\infty} = \sum_{t =0}^{\infty} (CA^t)^\top(CA^t), \label{ob_gram}
\end{align}
where $O_t$ is defined in (\ref{Obmatrix}). Note that $\lambda_{\max}(O_t^\top O_t)$ is non-decreasing with $t \in \bZ_+$, and for the asymptotically stable system, $\cO_{\infty}$ is finite. Now, we obtain the following result as a corollary of Theorem~\ref{dpr:thm}. 

\begin{seccor}\label{dpr_sta:thm}
The Gaussian mechanism (\ref{Out_noise}) induced by an asymptotically stable system (\ref{sys}) and $W_t \sim \cN_{(t+1)q}(\mu,\Sigma)$ is $(\varepsilon,\delta)$-differentially private for ${\rm Adj}_2^c$ at a finite time $t\in \bZ_+$ with $\varepsilon > 0$ and $1/2 > \delta > 0$ if the covariance matrix $\Sigma \succ 0$ is chosen such that the following inequality holds
\begin{align}
\lambda_{\min}^{1/2}(\Sigma ) \ge
c  \left(\lambda_{\max}^{1/2}(\cO_{\infty}) +\gamma \right) \cR (\varepsilon,\delta).\label{ep_Hinf}
\end{align}
\end{seccor}
\begin{proof}
It holds that
\begin{align*}
&|O_t (x'_0 - x_0) + N_t (U'_t - U_t) |_{ \Sigma^{-1}} \\
&\le |O_t (x'_0 - x_0)|_{ \Sigma^{-1}} + |N_t (U'_t - U_t) |_{ \Sigma^{-1}}\\
&\le \lambda_{\max}^{1/2}(\Sigma^{-1})(|O_t (x'_0 - x_0)|_2 + |N_t (U'_t - U_t) |_2)\\
& \le c \lambda_{\max}^{1/2}(\Sigma^{-1})  \left( \lambda_{\max}^{1/2}(\cO_{\infty}) + \gamma \right).
\end{align*}
Therefore, (\ref{ep_Hinf}) implies (\ref{z,delta}), where $1/\lambda_{\max}(\Sigma^{-1}) = \lambda_{\min}(\Sigma)$ is used.
\end{proof}

If $x_0$ is public and the multivariate Gaussian is i.i.d, Corollary~\ref{dpr_sta:thm} reduces to~\cite[Corollary~1]{NP:14}. When the initial state (resp. input sequence) is public,~the condition~\eqref{ep_Hinf} can be replaced by $\lambda_{\min}^{1/2}(\Sigma ) \ge c  \gamma \cR (\varepsilon,\delta)$ (resp. $\lambda_{\min}^{1/2}(\Sigma ) \ge c  \lambda_{\max}^{1/2}(\cO_{\infty}) \cR (\varepsilon,\delta)$). From the proof, one notices that for an asymptotically stable system~(\ref{sys}), if the covariance matrix $\Sigma$ is chosen such that (\ref{ep_Hinf}) holds, then (\ref{ep}) holds for any $t \in \bZ_+$. That is, for any asymptotically stable system~\eqref{sys} and for any $\varepsilon > 0$ and $1/2 > \delta > 0$, there exists a non-degenerate multivariate Gaussian noise which makes the induced mechanism $(\varepsilon,\delta)$-differentially private for any $t \in \bZ_+$. However, this is not always true for unstable systems; a similar statement can be found in~\cite[Theorem~4.5]{CDE:16}.

\subsection{Connection with Strong Input Observability}~\label{IOIID:s}
In the previous subsection, we have studied the $(\varepsilon,\delta)$-differential privacy of a Gaussian mechanism induced by output noise. However, it is not intuitively clear how differential privacy relates to dynamical systems' other intrinsic properties. For differential privacy, noise is designed to prevent the initial state and input sequence from being identified from the published output sequence. From the systems and control point of view, the property of determining the initial state and input sequence can be interpreted as observability or left invertibility~\cite{HP:98,SM:69}. In this subsection, we study the Gaussian mechanism from the input observability perspective.

First, we define what we mean by strong input observability.
\begin{secdefn}\label{sio:def}
The system~(\ref{sys}) is said to be \emph{strongly  input observable} if there exists $T \in \bZ_+$ such that both the initial state~$x_0\in \bR^n$ and initial input $u(0) \in \bR^m$ can be uniquely determined from the measured output sequence $Y_T$.
\red
\end{secdefn}

It is worth mentioning that if $(x_0,u(0))$ is uniquely determined from $Y_T$, then $(x(k),u(k))$ is consequently uniquely determined from $Y_{T+k}$,~$k=1,2,\dots$. Hence, one can focus on $(x_0,u(0))$ in the definition of strong input observability. Note that although strong input observability may seem too strong to hold for many existing engineering systems, more emerging and future systems may very likely possess this property after more sensed data and communicated information become available.

\begin{secrem}
There are several similar but different concepts from strong input observability just defined. On the one hand, if $U_T$ is known, the analysis reduces to determining the initial state $x_0$, i.e, the standard observability analysis~\cite{Kailath:80}. When $U_T$ is unknown, the property that $x_0$ can be uniquely determined is called unknown-input (or strong) observability~\cite{Kratz:95}. On the other hand, if $x_0$ is known, the analysis reduces to determining the initial input $u(0)$; this property is called input observability with the known initial state $x_0$~\cite{HP:98} or left invertibility~\cite{SM:69}. In the case, for the unknown initial state~$x_0$, the property that the initial input $u(0)$ can be uniquely determined is called input observability~\cite{HP:98}. Therefore, our strong input observability requires both unknown-input (or strong) observability and input observability. 
\red
\end{secrem}

The results in the existing observability analysis are helpful for the strong input observability analysis. Especially, by extending~\cite[Theorem~3]{SM:69}, we have the following necessary and sufficient condition for strong input observability. Since the proof is similar, it is omitted.
\begin{secthm}\label{rank:thm}
The system~(\ref{sys}) is strongly input observable if and only if 
\begin{align}
{\rm rank}\left[\begin{array}{cc}
O_{2n} & N_{2n,n}
\end{array}\right] = n+ (n+1)m \label{Kal_iob}
\end{align}
for $O_t$ in (\ref{Obmatrix}) and the submatrix $N_{t,T}$ of $N_t$ in (\ref{Markov}), i.e., the matrix $[O_{2n} \; N_{2n,n}]$, has the column full rank. 
\red
\end{secthm}

The following corollary is also used in this paper.
\begin{seccor}\label{rank:cor}
The system~(\ref{sys}) is strongly input observable if and only if 
\begin{align}
{\rm rank}\left[\begin{array}{cc}
O_t & N_{t,T}
\end{array}\right] = n + (T+1)m, \label{Kal_iob2}
\end{align}
for any integers $T \ge n$ and $t \ge T + n$.
\red
\end{seccor}
\begin{proof}
From the structures of $O_t$ and $N_{t,T}$, if $[O_{2n} \; N_{2n,n}]$ has the column full rank, then
\begin{align*}
{\rm rank}\left[\begin{array}{cc}
O_{2n} & N_{2n,n}
\end{array}\right]=	
{\rm rank}\left[\begin{array}{cc}
O_{2n+t} & N_{2n+t,n}
\end{array}\right] 
\end{align*}
for any $t \in \bZ_+$. Conversely, from the Cayley-Hamilton theorem~\cite{Lang:02}, if $[O_{2n+t} \; N_{2n+t,n}]$ has the column full rank  for some $t \in \bZ_+$, then (\ref{Kal_iob}) holds.
\end{proof}

The rank condition (\ref{Kal_iob}) or (\ref{Kal_iob2}) is a \emph{qualitative} criterion for strong input observability, but differential privacy is a \emph{quantitative} criterion. A connection between these two concepts can be established by extending the concept of the observability Gramian to strong input observability because controllability and observability Gramians give both quantitative and qualitative criteria. To extend the concept of the Gramian, we consider a weighted least square estimation problem\footnote{Note that the controllability Gramian is originally obtained from the minimum energy control problem~\cite{Kalman:60}. The duals of the controllability Gramian and minimum energy control problem are respectively the observability Gramian and least square estimation problem of the initial state.} of the initial state $x_0$ and input sequences $U_T$, $T \ge n$, from the output sequence with the measurement noise $Y_{w,t}$, $t \ge T + n$, under the technical assumption $u(\tau) =0$, $t \ge \tau > T$:
\begin{align}
J_{(x_0,U_T)} = \min_{(x_0,U_T) \in \bR^n \times \bR^{(T+1)m}} | Y_{w,t} - O_t x_0 - N_{t,T} U_T |_{\Sigma^{-1}}^2. \label{lstEs}
\end{align}
This problem has a unique solution if (\ref{Kal_iob2}) holds, i.e., the system is strongly input observable, in which case the solution is
\begin{align}
\left[\begin{array}{c}
\hat x_0 \\ \hat U_T
\end{array}\right] 
= \left(\cO_{\Sigma,t,T}\right)^{-1} \left[\begin{array}{cc}
O_t & N_{t,T}
\end{array}\right]^\top \Sigma^{-1} Y_{w,t}, \label{sol_least}
\end{align}
where
\begin{align}
\cO_{\Sigma,t,T}:=\left[\begin{array}{cc}
O_t & N_{t,T}
\end{array}\right]^\top \Sigma^{-1}\left[\begin{array}{cc}
O_t & N_{t,T}
\end{array}\right].\label{inpG}
\end{align}
When there is no measurement noise, i.e., $W_T=0$, it follows that (\ref{sol_least}) gives the actual initial state and input sequence. 

One notices that~$\cO_{\Sigma,t,t}=\cO_{\Sigma,t}$ for~$\cO_{\Sigma,t}$ in~\eqref{inpGtt}. As for~$\cO_{\Sigma,t}$, the matrix~$\cO_{\Sigma,t,T}$ characterizes the differential privacy level of a Gaussian mechanism, which we state as a corollary of Theorem~\ref{dpr:thm} without the proof.

\begin{seccor}\label{dpr:cor}
Let $T \ge n$ and $t \ge T+n$. For any~$((x_0,U_t),(x'_0,U'_t))$ belonging to ${\rm Adj}_2^c$ and satisfying $u(\tau)=u'(\tau)$, $T< \tau \le t$, the Gaussian mechanism (\ref{Out_noise}) induced by $W_t \sim \cN_{(t+1)q}(\mu,\Sigma)$ is $(\varepsilon,\delta)$-differentially private at a finite time $t\in \bZ_+$ with $\varepsilon > 0$ and $1/2 > \delta > 0$, if the covariance matrix $\Sigma \succ 0$ is chosen such that
\begin{align}
\lambda_{\max}^{-1/2}\left(\cO_{\Sigma,t,T} \right) \ge c \cR (\varepsilon,\delta).  \label{ep_inpG}
\end{align}
\red
\end{seccor}

Notice that if $T=t$, (\ref{ep_inpG}) is equivalent to (\ref{ep}). From~(\ref{ep_inpG}), Corollary~\ref{dpr:cor} concludes that the differential privacy of the Gaussian mechanism is characterized by the maximum eigenvalue of the matrix $\cO_{\Sigma,t,T}$, where $\cO_{\Sigma,t,T}$ is not necessarily non-singular in differential privacy analysis; non-singularity is required to guarantee the uniqueness of a solution to the least square estimation problem~(\ref{lstEs}). 

For $\Sigma=I_{(t+1)q}$, we call $\cO_{t,T}:=\cO_{I_{(t+1)q},t,T}$ the \emph{strong input observability Gramian}. The strong input observability Gramian is both qualitative and quantitative for strong input observability. For instance, from Corollary~\ref{rank:cor}, the system~(\ref{sys}) is strongly input observable if and only if $\cO_{t,T}$ is non-singular for any integers $T \ge n$ and $t \ge n +T$. Also, by substituting $(\hat x_0, \hat U_T)$ of~(\ref{sol_least}) into $(x_0, U_T)$ of (\ref{lstEs}), one notices that if all eigenvalues of $\cO_{t,T}$ is large, then $J_{(x_0,U_T)}$ in (\ref{lstEs}) with $\Sigma =I_{(t+1)q}$ is small. That is, $(x_0,U_T)$ is relatively easy to be estimated. This observation agrees with (\ref{ep_inpG}) because for $\Sigma=\sigma^2 I_{(t+1)q}$, large $\sigma$ is required if $\lambda_{\max}(\cO_{t,T})$ is large; recall~\eqref{sigma}. In other words, small noise is enough to make the less input observable Gaussian mechanism highly differentially private.

To gain deeper insight following the privacy analysis, we take a further look at the eigenvalues of the strong input observability Gramian $\cO_{t,T}$ from three aspects. First, from (\ref{Obmatrix}), (\ref{Markov}) and (\ref{inpG}) with $\Sigma=I_{(t+1)q}$, the first $m \times m$ block diagonal element of $\cO_{t,T}$ is
\begin{align*}
(\cO_{t,T})_{1,1} := \sum_{k=0}^t (CA^k)^\top CA^k,
\end{align*}
and for~$i \ge 2$, the $i$th $m \times m$ block diagonal element of $\cO_{t,T}$ is
\begin{align*}
&(\cO_{t,T})_{i,i} := D^\top D + \sum_{k=0}^{t-i-2}(CA^kB)^\top CA^kB, \\ 
&i=2,\dots,T +1
\end{align*}
where $(\cO_{t,T})_{T+1,T+1}:=D^\top D$ when $t=T$. One notices that $(\cO_{t,T})_{1,1}$ is the standard observability Gramian for the initial state~$x_0$, and $(\cO_{t,T})_{i,i}$,~$i \ge 2$ can be viewed as the observability Gramian corresponding to the initial input $u(0)$, which we call the \emph{initial input observability Gramian}. Since the trace of a matrix is the sum of all its eigenvalues, and the trace of $\cO_{t,T}$ is the sum of the traces of all its block diagonal elements $(\cO_{t,T})_{i,i}$, $i=1,\dots,T+1$, the sum of the eigenvalues of $\cO_{t,T}$ is the sum of the eigenvalues of all $(\cO_{t,T})_{i,i}$, $i=1,\dots,T+1$. Therefore, if the standard and initial input observability Gramians have large eigenvalues, the strong input observability Gramian $\cO_{t,T}$ has large eigenvalues also. In other words, the privacy level of the initial state and whole input sequence is characterized by that of only the initial state and initial input. This fact is natural because of two facts: 1) the output at each time instant contains the information of the initial state and initial input, i.e. these are the least private information; 2) if the initial state and initial input are uniquely determined, the whole input sequence is uniquely determined.

Next, for fixed $t$, the minimum eigenvalue of $\cO_{t,T}$ does not increase with $T$. For instance,
\begin{align}
\lambda_{\min}(\cO_{t,1}) \le \lambda_{\min}(\cO_{t,0}). \label{min}
\end{align}
Recall that these two Gramians are obtained from the least square estimation problems when $u(t)=0$ for $t=2,3,\dots$ and $t=1,2,\dots$, respectively. Therefore, (\ref{min}) corresponds to a natural observation that $u(0)$ is more difficult to estimate if $u(1)$ is unknown compared to the case when $u(1)$ is known to be $0$. 

Finally, for fixed~$T$, $\lambda_{\max}(\cO_{t,T})$ is non-decreasing with $t$, and thus $\varepsilon$ in Corollary~\ref{dpr:cor} is non-decreasing with $t$. This implies that as more data are being collected, less private a mechanism becomes. It is worth emphasizing that this observation is obtained when $\Sigma=I_{(t+1)q}$, or more generally $\Sigma= \sigma^2 I_{(t+1)q}$, $\sigma>0$, i.e., the output noise is i.i.d. Therefore, by employing non-i.i.d. noise, it is still possible to keep the same privacy level in longer duration; we will discuss this in the next subsection.

The above discussions are based on the minimum or maximum eigenvalue of the strong input observability Gramian. For more detailed privacy (strong input observability) analysis, each eigenvalue and the associated eigen-space can be used as typically done for the standard observability Gramian. Let $v_i \in \bR^{n+ (T+1)m}$, $i=1,\dots,n+ (T+1)m$, be the eigenvectors of $\cO_{t,T}$ associated with the eigenvalues $\lambda_i \le \lambda_{i+1}$. If there is $k$ such that $\lambda_k \ll \lambda_{k+1}$, then $(x_0,U_T) \in {\rm span}\{v_{k+1},\dots,v_T\}$ is relatively easy to observe. Especially, if $0<\lambda_{k+1}$, then such $(x_0,U_T)$ can be uniquely determined, and the projection of ${\rm span}\{v_{k+1},\dots,v_T\}$ onto the $(x_0,u(0))$-space gives the strongly input observable subspace. For the (non-strong) input observability with known initial state (i.e., left invertibility), the input observable and unobservable subspaces have been studied based on an extension of Kalman's canonical decomposition~\cite{SS:87}, but quantitative analysis has not been established yet.

The quantitative analysis of subspaces can be used for designing noise to make a system more private. Let $\lambda_k \ll \lambda_{k+1}$, and consider the projection of ${\rm span}\{v_{k+1},\dots,v_T\}$ onto the $(x_0,u(0))$-space, which we denote by $\cX \times \cU \subset \bR^n \times \bR^m$. Then, the output of the system is sensitive to the initial states and inputs in $\cX \times \cU$; in other words, such initial states and inputs are less private. To protect less private input information, one can directly add noise $v \in \cX \times \cU$ to the initial state and the input channel instead of the output channel. This motivates us to study differential privacy with input noise.

\subsection{Differential Privacy With Input Noise}\label{DPI:s}
In this subsection, we study the scenario where noise is added to the input channel. In this case, one can directly decide the distribution of estimated input data. However, additional effort is needed for studying the utility of the output data. Furthermore, differential privacy analysis is technically more involved because the output variables are not necessarily non-degenerate (while they are Gaussian if the input noise is Gaussian). To address this issue, even though artificial, some technical procedure is required, which is essentially equivalent to selecting a different base measure using the disintegration theorem~\cite{Billingsley:08}. As the main result of this subsection, we show that the differential privacy levels of the Gaussian mechanisms induced by the input noise and output noise can be made the same for suitable choices of the input noise and output noise.

To proceed with analysis, we assume that the system~(\ref{sys}) is strongly input observable, i.e., the matrix in (\ref{Kal_iob2}) has the column full rank for any $T \ge n$ and $t \ge T + n$, which implicitly implies $(t + 1)q \ge n + (T+1)m$. Then, there exists a $(t+ 1)q  -( n + (T+1)m)$ by $(t + 1)q$ matrix $\overline N_{t,T}$ such that
\begin{align*}
{\rm rank} \; \overline N_t = (t + 1)q, 
\end{align*}
and
\begin{align}
\left[\begin{array}{cc}
O_t & N_{t,T}
\end{array}\right]^\top \overline N_{t,T} = 0, \label{N2}
\end{align}
where
\begin{align}
\overline N_t:=\left[\begin{array}{ccc} O_t & N_{t,T} & \overline N_{t,T} \end{array}\right]. \label{N1}
\end{align}
\begin{secrem}
If a system is strongly input unobservable, i.e., (\ref{Kal_iob2}) does not hold, then one can use the singular value decomposition of $[O_t \; N_{t,T}]$ for similar analysis.
\red
\end{secrem}

Now, we consider the following system with the initial state, input and output noises, 
\begin{align}
\hspace{-3mm}\left\{\begin{array}{l}
x(t+1) = A x(t) + B (u(t) + v(t)), x(0) = x_0 + v_x\\
y_v(t) = C x(t) + D (u(t) + v(t)) + v_d(t), 
\end{array}\right.
\label{input/output_noise_sys}
\end{align}
where the output noise $v_d$ is generated by the dummy variables~$\overline V_{d,t,T} \in \bR^{(t + 1)q  -(  n + (T+1)m)}$ as
\begin{align}
\left[ \begin{array}{cccc} v_d (0); & v_d (1); &\cdots ;& v_d(t) \end{array}\right] = \overline N_{t,T} \overline V_{d,t,T}. \label{hat_w}
\end{align}
The reason we call them the dummy variables is that $\overline V_{d,t,T}$ does not affect the differential privacy level, which will be explained later. By recalling the notation of a sequence introduced in the introduction, define 
\begin{align}
\overline V_t :=\left[\begin{array}{ccc}v_x; &V_t; & \overline V_{d,t,T} \end{array}\right] \in \bR^{(t+1)q}. \label{Vt}
\end{align}
From (\ref{N1}) and (\ref{Vt}), for $v(\tau)=0$, $\tau > T$, the output sequence $Y_{v,t} \in \bR^{(t+1)q}$ can be described by
\begin{align}
Y_{v,t}  &= O_t (x_0 + v_x) + N_t (U_t +V_T)  + \overline N_{t,T} \overline V_{d,t,T}\nonumber\\
&= O_t x_0 + N_t U_t + \overline N_t \overline V_t.\label{input_noise}
\end{align}

We study the connection between the differential privacy levels of mechanisms~(\ref{Out_noise}) and~(\ref{input_noise}). The important fact is that the numbers of the elements of $W_t$ and $\overline V_t$ are the same, and from (\ref{N1}), $\overline N_t$ is non-singular. For mechanisms~(\ref{Out_noise}) and~(\ref{input_noise}), the generated output sequences  are the same if and only if $W_t = \overline N_t \overline V_t$. Therefore, the designs of the noises $W_t$ and $\overline V_t$ are equivalent problems. In the previous subsection, we have studied the differential privacy of the Gaussian mechanism~(\ref{Out_noise}). Similarly, for the Gaussian mechanism~(\ref{input_noise}), we have the following corollary of Theorem~\ref{dpr:thm}.

\begin{seccor}\label{dpr_dummy:thm}
Let $T \ge n$ and $t \ge T+n$. Also let $\overline V_t \sim \cN_{(t+1)q}(\overline \mu,{\rm diag}\{\overline \Sigma_1, \overline \Sigma_2\})$ be a non-degenerate multivariate Gaussian noise, where $\overline \Sigma_1 \in \bR^{(n+(T+1)m) \times (n+(T+1)m)}$ is the covariance matrix of the initial state and input noise $[v_x;V_t]$, and $\overline \Sigma_2$ is that of the dummy variable~$\overline{V}_{d,t,T}$. Then, for any~$((x_0,U_t),(x'_0,U'_t))$ belonging to ${\rm Adj}_2^c$ and satisfying $u(\tau)=u'(\tau)$, $T< \tau \le t$, the Gaussian mechanism (\ref{input_noise}) induced by the strongly input observable system (\ref{sys}) and $\overline V_t$ is $(\varepsilon,\delta)$-differentially private at a finite time $t\in \bZ_+$ if the covariance matrix $\overline \Sigma_1$ is chosen such that
\begin{align}
\lambda_{\min}^{1/2}(\overline \Sigma_1) \ge c \cR (\varepsilon,\delta).  \label{epinoise}
\end{align}
\end{seccor}

\begin{proof}
Instead of (\ref{ep}), one has
\begin{align*}
&\lambda_{\max}^{-1/2}\left(\left[\begin{array}{cc}
O_t & N_{t,T}
\end{array}\right]^\top \overline N_t^{-\top} \overline \Sigma^{-1} \overline N_t^{-1} \left[\begin{array}{cc}
O_t & N_{t,T}
\end{array}\right]\right)\\
&  \ge c \cR (\varepsilon,\delta).
\end{align*}
From (\ref{N1}), it follows that
\begin{align*}
&\left[\begin{array}{cc}
O_t & N_{t,T}
\end{array}\right]^\top \overline N_t^{-\top} \overline \Sigma^{-1} \overline N_t^{-1} \left[\begin{array}{cc}
O_t & N_{t,T}
\end{array}\right]\\
&=
\left[\begin{array}{cc}
I_{n+(T+1)m} & 0
\end{array}\right]
\overline N_t^\top \overline N_t^{-\top} \overline \Sigma^{-1} \overline N_t^{-1} \overline N_t \left[\begin{array}{c}
I_{n+(T+1)m} \\ 0
\end{array}\right]\\
&= \overline \Sigma_1^{-1}.
\end{align*}
Therefore, (\ref{epinoise}) holds.
\end{proof}

Corollary~\ref{dpr_dummy:thm} concludes that the differential privacy level only depends on the covariance $\overline \Sigma_1$ of the input noise $[v_x;V_T]$, i.e., the differential privacy level does not depend on the system itself. The covariance $\overline \Sigma_1$ gives an intuitive interpretation of the privacy level of the input. Therefore, Corollary~\ref{dpr:cor} can help understanding the interpretation of the magnitudes of $(\varepsilon,\delta)$ from the perspective of the privacy level of the input. 

In Corollary~\ref{dpr:cor} and Theorem~\ref{dpr_dummy:thm}, the differential privacy levels of both mechanisms are the same if
\begin{align}
\cO_{\Sigma,t,T} = \left[\begin{array}{cc}O_t &N_t \end{array}\right]^\top \Sigma^{-1}\left[\begin{array}{cc}O_t &N_t\end{array}\right] = \overline \Sigma_1^{-1}, \label{dis_input/ouput_noise}
\end{align}
where we recall (\ref{inpG}) for the first equality; the converse is not true in general since differential privacy only evaluates the maximum eigenvalues. Therefore, adding the Gaussian noise with the covariance $\Sigma$ to the output of the system~(\ref{sys}) is equivalent to adding the Gaussian noise with the covariance $\cO_{\Sigma,t,T}^{-1}$ to the input of the system (\ref{sys}) under the strong input observability assumption. 

In the previous subsection, we mentioned that the privacy level of the mechanism~(\ref{Out_noise}) with the i.i.d. output noise $\Sigma = \sigma I_{(t+1)q}$ decreases with the growth of duration. In contrast, if one adds noise to the initial state and input channel, the privacy level of a mechanism does not depend on the duration because one can directly decide the distribution of the estimated initial state and input sequence. These two facts do not contradict each other if one allows to add non-i.i.d output noise. From~(\ref{dis_input/ouput_noise}), adding suitable non-i.i.d. noise to the output channel has a similar effect as adding noise to the initial state and input channel. Therefore, adding non-i.i.d. noise is a key factor for keeping the same privacy level against the duration when one adds noise to the output channel.

Finally, the reason that the dummy variables $\overline V_{d,t,T}$ do not affect the differential privacy level can be explained based on the least square estimation problems of the initial state and input sequence. For a strongly input observable system, the solution to the following least square estimation problem
\begin{align*}
\overline J_{(x_0,U_T)} = \min_{(x_0,U_T) \in \bR^n \times \bR^{(T+1)m}} | Y_{v,t} - O_t x_0 - N_{t,T} U_T |_2^2
\end{align*}
is,  from (\ref{Kal_iob2}), (\ref{N2}), and (\ref{input_noise}),
\begin{align*}
\left[\begin{array}{c}
\hat x_0 \\ \hat U_T
\end{array}\right] 
&= \cO_{t,T}^{-1}\left[\begin{array}{cc}
O_t & N_{t,T}
\end{array}\right]^\top  Y_{v,t}
= \left[\begin{array}{c}
x_0 \\ U_T
\end{array}\right] + \left[\begin{array}{c}
v_x \\ V_T
\end{array}\right].
\end{align*}
The least square estimation is the actual initial state and input sequence plus the noise added to them. Because of the condition (\ref{N2}), the dummy variable $\overline V_{d,t,T}$ is canceled. This is the reason that the dummy variable does not affect differential privacy analysis.

\section{Privacy-Preserving Controllers}\label{PCD:s}
\subsection{Motivating Example}~\label{Mot:s}
We start with a motivating example. Consider DC microgrids~\cite{CTP:18} installed with smart meters whose dynamics are described by
\begin{align}
L_i  \dot I_i(t) &= - R_i I_i(t)  - V_i(t) + u_i(t), \ I_i(t):=I_{T,i}(t)-I_{L,i}, \nonumber\\
C_i \dot V_i(t) &= I_i(t) - \sum_{j \in {\cal N}_i} I_{i,j}(t), \nonumber\\
L_{i,j} \dot I_{i,j}(t) &= (V_i(t) -V_j(t)) - R_{i,j} I_{i,j}(t),\nonumber\\
y_{i,1}(t) &= V_i(t),\  y_{i,2}(t) = I_i(t), \label{ex:ps}
\end{align}
where~$I_{T,i} (t) \in \bR$,~$V_i(t) >0$, and~$I_{i,j}(t) \in \bR$ denote the generator current, load voltage, the current between nodes~$i$ and~$j$, respectively, and~$I_{L,i} \in \bR$ denote the load current, which can be viewed as a constant in the time scale of controller design. The parameters~$L_i,L_{i,j},R_i,R_{i,j},C_i>0$ denote inductances, resistances, and capacitance, respectively. The set of neighbors of node~$i$ is denoted by~${\cal N}_i$, and the number of the neighbors is denoted by~$n_i$. For analysis and controller design, we use its zero-order-hold discretization, since each output information is collected and sent to the power company digitally.

One objective of the power company is to maintain the stability of the system by keeping~$V_i(t)$ to the prescribed value~$V^*$ and the difference between the generator current (i.e. supply) and load current (i.e. demand), denoted by~$I_i(t)$, to zero. Therefore, the control objective is
\begin{align}
\lim_{t\to\infty} V_i(t) = V^*, \ \lim_{t\to\infty} I_i(t) = 0.
\label{ex:objective}
\end{align}

Owing to developments of IoT technologies, smart meters are becoming more widely available, which can be used to monitor and send the value of~$I_i(t)(=I_{T,i}(t)-I_{L,i})$ to the power company online. However, the desired load current~$I_{L,i}$ is determined by each user and thus contains the information of each user's lifestyle. Since this load current of privacy concern is static, one can use existing results for static differential privacy, e.g.~\cite{SDT:15}. 

However, there is bigger privacy issue that needs to be addressed. Our observations in the previous section indicate the possibility that a user $i$ can identify the other users' $[V_i, I_i]$ from its own dynamical control input data sets~$u_i$. So the privacy of user $i$ here is concerned with her wish not letting the other users be able to identify that her consumption pattern has changed, and such a privacy issue depends on controller dynamics. Thus, one is forced to consider designing a tracking controller by taking privacy into account. The privacy-protection objective is that even if user $i$'s $[V_i, I_i]$ becomes different from~$[V^*,0]$, another user $j$ cannot infer the occurrence of the difference from~$u_j$, $j \neq i$. This privacy requirement should not conflict with the control objective~\eqref{ex:objective} of tracking the desired signals.

In the following subsections, first we summarize the standard result for tracking controller design based on the internal model principle. Then, we impose a differential privacy requirement for a tracking controller. In the end, we consider estimating private information and evaluate its difficulty.

\subsection{Tracking Controllers}
To be self-contained, in this subsection, 
an existing tracking controller is shown. This controller has tuning parameters that will be adjusted based on a privacy requirement in the next subsection.

Consider the following plant
\begin{align}
\left\{\begin{array}{l}
x_p(t+1) = A_p x_p(t) + B_p u_p(t),\\
y_p(t) = C_p x_p(t) + D_p u_p(t), 
\end{array}\right. \label{sys_p}
\end{align}	
where $x_p(t)\in\bR^{n_p}$, $u_p(t) \in \bR^{m_p}$ and $y_p(t) \in \bR^{q_p}$ denote the state, input and output, respectively, and $A_p \in \bR^{n_p\times n_p}$, $B_p \in \bR^{n_p \times m_p}$, $C_p \in \bR^{q_p \times n_p}$ and $D_p \in \bR^{q_p \times m_p}$. 

The control objective is to design an output feedback controller, which achieves $y_p \to y_r$ as $t \to \infty$ for a given reference output $y_r(t) \in \bR^{q_p}$. Suppose that the reference output $y_r(t)$ is generated by the following exosystem:
\begin{align}
\left\{\begin{array}{l}
x_r(t+1) = A_r x_r(t), \ x_r(0) = x_{r,0} \in \bR^{n_r},\\
y_r(t) = C_r x_r(t), 
\end{array}\right. \label{sys_ref}
\end{align}
where $x_r(t)\in\bR^{n_r}$ and $y_r(t) \in \bR^{q_r}$; $A_r \in \bR^{n_r \times n_r}$ and $C_r \in \bR^{q_r \times n_r}$. Then, the composite system consisting of the plant (\ref{sys_p}) and exosystem (\ref{sys_ref}) is
\begin{align*}
&\hspace{-5mm}\left\{\begin{array}{l}
\bar x(t+1) = \bar A \bar x(t) + \bar B u_p(t),\\
e(t) = y_p(t) - y_r(t) = \bar C \bar x(t) + D_p u_p(t),
\end{array}\right.\\
&\bar x:=\left[\begin{array}{c}
x_p \\ x_r
\end{array}\right], \
\bar A:=\left[\begin{array}{cc}
A_p & 0 \\ 0 & A_r
\end{array}\right], \
\bar B:=\left[\begin{array}{cc}
B_p  \\ 0
\end{array}\right],\\
&\bar C:=\left[\begin{array}{cc}
C_p & - C_r 
\end{array}\right].
\end{align*}
The tracking control objective can be rewritten as $\lim_{t \to \infty} e(t) = 0$.  

As an output feedback controller, the following observer based stabilizing controller is typically used
\begin{align}
\left\{\begin{array}{l}
u_p(t) = G x_c(t),\\
x_c(t+1) = A_c x_c(t) - L e (t), 
\end{array}\right.  \label{sys_con}
\end{align}
where
\begin{align*}
A_c:= \bar A + L \bar C + (\bar B + L D_p)G,
\end{align*}
and $G=[G_1 \ G_2] \in \bR^{m_p \times (n_p + n_r)}$ and $L =[L_1^\top \ L_2^\top]^\top \in \bR^{(n_p + n_r) \times q_p}$ are design parameters. The tracking problem is solvable by the above dynamic output feedback controller under the following standard assumptions~\cite{Huang:04}.
\begin{secass}\label{reg:ass1}
The matrix $A_r$ has no eigenvalue in the interior of the unit circle.
\red
\end{secass}
\begin{secass}\label{reg:ass2}
The pair $(A_p,B_p)$ is stabilizable.
\red
\end{secass}
\begin{secass}\label{reg:ass3}
The pair $(\bar C, \bar A)$ is detectable.
\red
\end{secass}
\begin{secass}\label{reg:ass4}
The following two equations:
\begin{align*}
X A_r &= A_p X + B_p U,\\
0 &= C_p X + D_p U - C_r,
\end{align*}
have a pair of solutions $X \in \bR^{n_p \times n_r}$ and $U \in \bR^{m_p \times n_r}$.
\red
\end{secass}

\begin{secrem}
Assumption~\ref{reg:ass4} guarantees that for any given $x_r(t)$ generated by (\ref{sys_ref}), there exist $x_{p,s}(t)$ and $u_{p,s}(t)$ simultaneously satisfying (\ref{sys_p}) and $e(t) = y_p(t) - y_r(t)  = 0$ for all $t \in \bZ_+$. Assumption~\ref{reg:ass1} guarantees that such $x_{p,s}(t)$ and $u_{p,s}(t)$ uniquely exist; this assumption is for the ease of discussion and is not necessarily to be imposed as mentioned in~\cite{Huang:04}. \red
\end{secrem}

Under Assumption~\ref{reg:ass4}, the tracking problem is solvable if the closed-loop system consisting of the plant~(\ref{sys_p}) and the controller~(\ref{sys_con}) is asymptotically stable. From the separation principle~\cite{Kailath:80}, the closed loop system can be made asymptotically stable by finding a pair of~$G_1$ and~$L$ that makes~$A_p + B_p G_1$ and~$\bar A + L \bar C$ asymptotically stable, respectively. Then,~$G_2$ can be designed as~$G_2 = U - G_1 X$ for~$U$ and~$X$ in Assumption~\ref{reg:ass4}. 

\subsection{Privacy Requirements for Controllers}
The privacy requirement imposed in the motivating example is to make a user $j$ not be able to distinguish whether another user $i$'s~$[V_i, I_i]$ has deviated from $[V^*,0]$ using its input~$u_j$,~$j \neq i$. This corresponds to designing a controller~\eqref{sys_con} such that~$e$ is always inferred to be zero using~$u_p$. Note that this privacy requirement is different from protecting the privacy of~$y_p$, in which case if $y_r$ is a piece of public information, the information $e = y_p - y_r = 0$ cannot be published, and thus in which case protecting~$y_p$ conflicts with the tracking control objective, implying one may have to regulate~$y_p$ to a different value than~$y_r$. In contrast, the privacy requirement for~$e$ does not conflict with the goal of tracking control.

For protecting the information of~$e$, we consider adding noise to~$u_p$. As mentioned in the previous section, adding sufficiently large noise always achieves the prescribed privacy level. However, large noise can change a control input significantly. Therefore, it is desirable to design a controller which becomes highly private by adding small noise. According to Theorem~\ref{dpr_sta:thm}, such a controller has a small $H_{\infty}$-norm. 

\begin{secrem}
One may consider controller design from different perspectives. Based on Theorem~\ref{dpr:thm}, differential privacy analysis itself is possible for an unstable controller. However, this theorem does not give a clear indication on how to choose design parameters $G_1$ and $L_1$. On the other hand, if a strongly input unobservable controller is designed, the information in the strongly input unobservable space is protected without adding noise as mentioned in Section~\ref{IOIID:s}. However, from Theorem~\ref{rank:thm}, this reduces to a rank constraint problem that is difficult to solve in general as the rank minimization problem is known to be NP-hard~\cite{VB:96}.  Therefore, we design a controller having a small $H_{\infty}$-norm.
\red
\end{secrem}

\begin{secrem}\label{gram:rem}
In Theorem~\ref{dpr_sta:thm}, the differential privacy level also depends on the standard observability Gramian of the initial state. However, it is not straightforward to simultaneously specify the maximum eigenvalues of the observability Gramian and $H_{\infty}$-norm. In fact, it is known that the maximum Hankel singular value, the square root of the maximum eigenvalue of the product of the controllability and observability Gramians, is bounded by the $H_{\infty}$-norm~\cite{ZDG:96}. Therefore, making $H_{\infty}$-norm small can result in making the maximum eigenvalue of the observability Gramian small.
\red 
\end{secrem}
\begin{secrem}
Even if one adds different noise than the Gaussian noise such as the Laplace noise as in Remark~\ref{Laplace:rem}, making $H_{\infty}$-norm small can increase the differential privacy level. Making $H_{\infty}$-norm small can result in making ${\lambda_{\max}^{1/2}([\begin{array}{cc} O_t & N_t \end{array}]^\top [\begin{array}{cc} O_t & N_t \end{array}])}$ small. Then, from the equivalence of the norm, any matrix induced norm of $[\begin{array}{cc} O_t & N_t \end{array}]$ becomes small. Therefore, from Remark~\ref{Laplace:rem}, the differential privacy level increases also for the Laplace mechanism.
\red
\end{secrem}

In general, a controller having a bounded $H_{\infty}$-norm needs to be asymptotically stable. Unfortunately, stable controller design is not always possible because of its structure in~(\ref{sys_con}).
\begin{secprop}
Under Assumptions~\ref{reg:ass1}-\ref{reg:ass4}, the controller~(\ref{sys_con}) solving the linear output regulation problem is not asymptotically stable if $D_p =0$.
\end{secprop}
\begin{proof}
Assumption~\ref{reg:ass4}, (\ref{sys_con}), and~$G_2 = U - G_1 X$ yield
\begin{align*}
&\left[\begin{array}{cc}
\lambda I_{n_p} - A_p - B_p G_1 &- B_p G_2 \\ 
0 & \lambda I_{n_r} - A_r\\
C_p  &  -C_r
\end{array}\right]
\left[\begin{array}{c}
X \\ 
I_{n-r}
\end{array}\right]\\
&=
\left[\begin{array}{c}
\lambda X - A_p X  - B_p U  \\ 
 \lambda I_{n_r} - A_r\\
C_p X  -C_r
\end{array}\right]
=
\left[\begin{array}{c}
X(\lambda I_{n_r} - A_r) \\ 
 \lambda I_{n_r} - A_r\\
 - D_p U
\end{array}\right].
\end{align*}
If $D_p =0$, this becomes zero when $\lambda$ is an eigenvalue of $A_r$. Therefore, for the pair~$(\bar C, \bar A + \bar B G)$, any eigenvalue of~$A_r$ is not observable. That is, the set of eigenvalues of $A_c$ contains that of $A_r$, which are marginally stable according to Assumption~\ref{reg:ass1}.
\end{proof}

If $D_p \neq 0$, one can use the output regulation controller~(\ref{sys_con}) addressing the privacy requirement. However, there are plenty of systems for which $D_p = 0$. To deal with these systems, we modify the output regulation controller~(\ref{sys_con}) in the next subsection.

\subsection{Controller Design with Privacy Concern}
In order to address the case $D_p=0$, we consider the following controller dynamics:
\begin{align}
&\left\{\begin{array}{l}
u_p(t) = G_1 \bar x_c(t) + G_2 x_r(t),\\
\bar x_c(t+1) = \bar A_c  \bar x_c(t) + \bar A_r  x_r(t) - L_1 e (t),
\end{array}\right.  \label{sys_con2}
\end{align}
where
\begin{align*}
&\bar A_c := A_p + B_p G_1 + L_1 (C_p + D_p G_1),\\
&\bar A_r := L_1 C_r + (B_p  + L_1 D_p ) G_2.
\end{align*}
The difference of~\eqref{sys_con2} from the previous controller~(\ref{sys_con}) is to use the actual state~$x_r$ of the exosystem~(\ref{sys_ref}) instead of its estimation. Since we do not need to estimate~$x_r$, \eqref{sys_con2} can have better control performance than~(\ref{sys_con}).

Privacy-preserving tracking controller design requires the following three conditions for the new controller parameters $G_1$ and $L_1$:
\begin{enumerate}
\item $A_p+ B_p G_1$ is asymptotically stable;
\item $A_p + L_1 C_p$ is asymptotically stable;
\item Given $\gamma>0$, the $H_{\infty}$-norm of the controller (\ref{sys_con2}) from $e$ to $u_p$ is bounded as 
\begin{align}
\| -G_1 (z I_{n_p + n_r} - \bar A_c)^{-1} L_1 \|_{H_{\infty}} \le \gamma. \label{hinf2}
\end{align}
\end{enumerate}
As mentioned in the previous subsection, the third condition implicitly requires the stability of the new controller~(\ref{sys_con2}). Stabilization of a plant by a stable controller is called strong stabilization. Its necessary and sufficient condition is described in terms of a parity interlacing property (PIP) of the transfer function matrix~\cite{Vidyasagar:11}. However, the PIP condition does not provide a controller design method. For continuous-time systems, the papers~\cite{MO:00,GO:05} provide ways of designing a controller satisfying Condition 3) based on the LMI. We employ one of these methods. 

It is not easy to simultaneously finding $G_1$ and $L_1$ satisfying all three conditions; the reason will be explained later. Therefore, first, we find $G_1$ stabilizing $A_p+ B_p G_1$, which can be done by multiple methods under Assumption~\ref{reg:ass2}. Then, we find $L_1$ satisfying 2) and 3) as follows. 
\begin{seclem}\label{LMI:thm}
Suppose that~$G_1$ is chosen such that~$A_p+ B_p G_1$ is asymptotically stable. If there exist $P \in \bR^{n_p \times n_p}$ and $\hat L_1 \in \bR^{n_p \times q_p}$ satisfying the following LMIs:
\begin{align}
\left[\begin{array}{cc}
P & P A_p + \hat L_1 C_p  \\
(P A_p + \hat L_1 C_p)^\top & P
\end{array}\right] \succ 0,\label{LMI2}
\end{align}
and
\begin{align}
&\left[\begin{array}{cccc}
P & 0 & \overline P_{13} & G_1^\top\\
0 & \gamma^2 I_{q_p} & -\hat L_1^\top & 0\\
\overline P_{13}^\top & -\hat L_1 & P &0\\
G_1 & 0 & 0 & I_{m_p}
\end{array}\right] \succ 0, \label{LMI}\\
&\overline P_{13}^\top := P ( A_p + B_p G_1 )  + \hat L_1 ( C_p  + D_p G_1 ),\nonumber
\end{align}
then $A_p + L_1 C_p$ with $L_1:=P^{-1} \hat L_1$ is asymptotically stable, and (\ref{hinf2}) holds. 
\end{seclem}
\begin{proof}
If~(\ref{LMI2}) holds,  $A_p + L_1 C_p$ is asymptotically stable. Next, (\ref{LMI}) implies (\ref{hinf2})~\cite[Theorem~4.6.6]{SIG:97}. 
\end{proof}
\begin{secrem}
For any given $G_1$ stabilizing $A_p + B_p G_1$, it is possible to verify if there exist $P$, $L_1$, and $\gamma>0$ satisfying~\eqref{hinf2} by replacing (\ref{LMI}) by
\begin{align}
\left[\begin{array}{cc}
P & \overline P_{13} \\
\overline P_{13}^\top & P 
\end{array}\right] \succ 0. \label{LMI12}
\end{align}
That is, given~$G_1$, the LMIs (\ref{LMI2}) and (\ref{LMI12}) have a solution~$P$ only if strong stabilization is achievable.
\red
\end{secrem}

An alternative way of controller design is to find~$\hat L_1$ satisfying~2) and then to use similar LMIs for finding~$G_1$ that satisfies~1) and~3) simultaneously. If one tries to find~$G_1$ and~$\hat L_1$ at the same time, then one encounters BMIs, e.g. there is a cross term of~$G_1$ and~$P$ or~$G_1$ and~$\hat L_1$ in~$\overline P_{13}$ in~(\ref{LMI}). BMIs are more difficult to handle than LMIs, since a BMI describes those sets that are not necessarily convex.

\subsection{Differential Privacy of Controllers}
To make the designed controller in the previous subsection differentially private, one can add noise to the output~$u_p$ or the input~$e$ of the controller. As clarified in Corollary~\ref{dpr_dummy:thm},  the differential privacy level under the input Gaussian noise only depends on the covariance matrix of the noise. Under the output Gaussian noise, we obtain the following theorem by combining Corollary~\ref{dpr_sta:thm} and Lemma~\ref{LMI:thm}. Since the proof directly follows, it is omitted. 
\begin{secthm}
Consider the controller dynamics~\eqref{sys_con2} satisfying the requirements~1) -- 3) with the output~$u_p(t)+w(t)$, where~$w(t) \in \bR^{m_p}$ is the noise. Then, the Gaussian mechanism induced by the controller dynamics and~$W_t \sim \cN_{(t+1)m_p}(\mu,\Sigma)$  is~$(\varepsilon,\delta)$-differentially private for~${\rm Adj}_2^c$ at a finite time~$t\in \bZ_+$ with~$\varepsilon > 0$ and~$1/2 > \delta > 0$ if the covariance matrix~$\Sigma \succ 0$ is chosen such that~(\ref{ep_Hinf}) holds for $(A,B,C,D)=(\bar A_c,-L_1,G_1,0)$.
\red
\end{secthm}

In summary, the privacy-preserving controller with the prescribed differential privacy level is designed as follows. First, one designs the controller dynamics (\ref{sys_con2}) based on the LMIs~(\ref{LMI2}) and~(\ref{LMI}) and then design the noise $w$ based on the above theorem with (\ref{ep_Hinf}). In the LMIs, the design parameters reduce to~$\gamma$, the $H_{\infty}$-norm of the controller (\ref{sys_con2}). 

From (\ref{ep_Hinf}) (and Remark~\ref{gram:rem}), a smaller~$\gamma$ gives a smaller lower bound on the covariance matrix of the Gaussian noise, but making~$\gamma$ small may result in deterioration of the control performance. Moreover, adding noise~$w$ may result in deterioration of the control performance also. 
Let~$H(z)$ and~$K(z)$ denote the transfer functions of the plant~\eqref{sys_p} from~$u_p$ to~$y_p$ and controller~\eqref{sys_con2} from~$e$ to~$u_p$, respectively. The transfer function matrices of the closed-loop system from~$w$ to~$y_p$ is~$(I-H(z)K(z))^{-1}P(z)$. If the controller is designed such that the~$H_\infty$-norm of~$K(z)$ is sufficiently large, the output~$y_p$ of the closed-loop system is less influenced by~$w$. In contrast, this causes a decrease in the privacy level. Therefore, there is a trade-off between the control performance and the privacy level for privacy-preserving controller design.

If one additionally requires the $H_{\infty}$-norm of the closed-loop system not to be greater than~$\bar \gamma>0$, then one can use the following LMI:
\begin{align}
&\left[\begin{array}{cccccc}
Q & 0 & 0 & *& *  & *\\
0 & P & 0 & * & * &*\\
0 & 0 & \bar \gamma^2 I_{q_p} & * &* & *\\
Q A_p & Q B_p G_1 & Q B_p &Q & 0& 0\\
- \hat L_1 C_p & \overline P_{25}^\top & - \hat L_1 D_p& 0 & P &0\\
C_p & D_p G_1 & D_p & 0 & 0 &  I_{m_p}
\end{array}\right] \succ 0, \label{LMI3}\\
&\overline P_{25}^\top =P ( A_p + B_p G_1 )  + \hat L_1 C_p,\nonumber
\end{align}
where $*$ are suitable elements to make the matrix symmetric. The~$H_{\infty}$-norms of the controller and closed-loop system are made less than~$\gamma$ and~$\bar \gamma$, respectively, if LMIs~(\ref{LMI2}),~(\ref{LMI}) and~\eqref{LMI3} have solutions~$P$,~$Q$, and~$\hat L_1$.

\subsection{Private Data Estimation}\label{PDE:s}
In the previous subsections, we have studied privacy-preserving controller design. An approach to evaluating  the privacy level of the proposed controller is to utilize differential privacy. In systems and control, filtering is a central problem, and one may ask whether existing filtering techniques can be used for estimating private data. Therefore, in this subsection, we consider this estimation problem. It is expected that the obtained observations in this subsection can help in improving the privacy-preserving controller design method.

For state estimation, one can use the standard techniques of the optimal linear filters or smoothers. Thus, we reformulate the input estimation problem as a state estimation problem inspired by unknown input observer design~\cite{Bhattacharyya:78,HM:92}. Suppose that the designed controller (\ref{sys_con2}) is strongly input observable for the output $u_p$ and input $e$. Recall the notations for sequences $U_{p,2n}(t)$ and $E_{2n}(t)$ introduced in the introduction. In a similar manner as (\ref{Outseq}), the output sequence $U_{p,t}$ of the controller can be described by
\begin{align}
U_{p,2n}(t) = O_{2n} \bar x_c (t)+ N_{2n} E_{2n}(t) + N_{r,2n} X_{r,2n}(t), \label{Outseq_con}
\end{align}
where $A = \bar A_c$, $B = -L_1$, $C=G_1$, and $D=0$ for $O_{2n}$ and $N_{2n}$, and $N_{r,t}$ denotes $N_t$ for $A = \bar A_c$, $B = \bar A_r$, $C=G_1$, and $D=G_2$.

From (\ref{Kal_iob}), there exists a (not necessarily unique) matrix $K \in \bR^{(n +(n+1)m) \times (2n+1)q}$ such that
\begin{align}
K \left[\begin{array}{cc}
O_{2n} & N_{2n,n}
\end{array}\right] = I_{n + (n+1) m}. \label{InpinvK}
\end{align}
By using this $K$, define
\begin{align*}
&K_x:=\left[\begin{array}{cc}
I_n & 0
\end{array}\right] K, \\ 
&K_u:=\left[\begin{array}{ccc}
0 & I_m & 0
\end{array}\right] K.
\end{align*}
Then, from (\ref{Outseq_con}), 
\begin{align}
&K_x (U_{p,2n} - N_{r,2n} X_{r,2n}) \nonumber\\
& = \left[\begin{array}{cc}
I_n & 0
\end{array}\right] K \left[\begin{array}{cc}
O_{2n} & N_{2n,n}
\end{array}\right] \left[\begin{array}{c}
\bar x_c(0) \\ E_n
\end{array}\right] \nonumber\\
&= \left[\begin{array}{cc}
I_n & 0
\end{array}\right]\left[\begin{array}{c}
\bar x_c(0) \\ E_n
\end{array}\right] = \bar x_c(0),
\end{align}
and
\begin{align}
&K_u (U_{p,2n}(t) - N_{r,2n} X_{r,2n}(t)) \nonumber\\
&= \left[\begin{array}{ccc}
0 & I_m & 0
\end{array}\right]\left[\begin{array}{c}
\bar x_c(t) \\ E_n(t)
\end{array}\right]= e(t), \label{linv_out}
\end{align}

By substituting them into (\ref{sys_con2}), we have
\begin{align}
&\left\{\begin{array}{l}
u_p(t) = G_1 \bar x_c(t) + G_2 x_r(t),\\
\bar x_c(t+1) = \bar A_c  \bar x_c(t) + \bar A_r x_r(t)\\
\hspace{18mm} - L_1 K_u (U_{p,2n}(t) - N_{r,2n} X_{r,2n}(t)) ,\\
\bar x_c(0) = K_x (U_{p,2n} - N_{r,2n} X_{r,2n}),
\end{array}\right. \label{sys_con_ob2}
\end{align}
where recall that the state of the exosystem $x_r$ is a piece of public information. This system corresponds to a left inverse system of the controller. In order to estimate~$e$ from~$u_p$, one can use the state estimation of this model with the process and measurement noises~$\tilde v(t) \in \bR^{(2n+1)m_p}$ and $\tilde w(t) \in \bR^{m_p}$. 

Let $\tilde x_c(t)$ denote the state estimation of (\ref{sys_con_ob2}). Then, define
\begin{align*}
\tilde u_p(t) = G_1 \tilde x_c(t) + G_2 x_r(t).
\end{align*}
Finally from~(\ref{linv_out}) and $\tilde U_{p,2n}(t)$, the estimation of $e(t)$ denoted by $\tilde e(t)$ can be computed by
\begin{align}
\tilde e(t) = K_u (\tilde U_{p,2n}(t) - N_{r,2n} X_{r,2n}(t)).
\end{align}

It is worth mentioning that in (\ref{sys_con_ob2}), future information of $u_p(t)$, namely $U_{p,2n}(t)$ is used in order to estimate $e(t)$. In other words, at time~$t$, one can estimate the historic data $e(t-2n)$, and thus the private data estimation can be formulated as a smoothing problem. There are several techniques for designing filters or smoothers such as the Kalman filter or its smoother, and one of them can be employed for the state estimation. Typically, for the filtering and smoothing problems,  i.i.d. Gaussian noises are used as the process and measurement noises. Therefore, adding non-i.i.d. or non-Gaussian noises to the privacy-preserving controller could be useful for protecting the private data than adding i.i.d. Gaussian noises.

The above is one approach to input data estimation. For strongly input observable systems, one can directly estimate~$(x_0,u(0))$ from~$E_t$ and the probability density function of noise by extending the results in~\cite{HCG:18}. The paper~\cite{HCG:18} further develops an updating algorithm of the estimation forward in time.

\section{Examples}\label{Ex:s}
We revisit the DC microgrids~\eqref{ex:ps} with parameters in~\cite{CTP:18} for~$i=1,2$, where~$R_i=0.2[\Omega]$,~$R_{i,j}=70[{\rm m\Omega}]$,~$L_i=1.8[{\rm mH}]$, and~$C_i=2.2[{\rm mF}]$ and design a privacy-preserving controller, where the sampling period is~$10^{-3}[{\rm s}]$. We consider that originally~$I_i=0[{\rm A}]$ and~$V_i=380[{\rm V}]$ are achieved with~$I_{1,2}=0[{\rm A}]$. Then the user~$1$ starts to use more electricity, which causes~$I_1=-4[{\rm A}]$. The goal is to achieve~$I_i=0[{\rm A}]$ and~$V_i=380[{\rm V}]$ again by protecting from user~$2$ the information that user~$1$ changes its electricity consumption. From the control objective~\eqref{ex:objective}, the exosystem~\eqref{sys_ref} is given by~$A_r=C_r=I_4$. In this problem setting, Assumptions~\ref{reg:ass1}-\ref{reg:ass4} hold. 

We design a privacy-preserving tracking controller. First, we design $G_1$ stabilizing $A_p + B_p G_1$ based on the following optimal control problem: 
\begin{align*}
J = \sum_{t = 0}^{\infty} |x_p(t)|_2^2 + |u_p(t)|_2^2.
\end{align*}
Solving the corresponding Riccati equation, $G_1$ is obtained as
\begin{align*}
G_1 =\left[\begin{array}{ccccc}
-0.850 & 0.037 & -0.0461 & -0.0007 & 0.229\\
0.0370 & -0.850 & -0.0007 & -0.0461 & -0.229
\end{array}\right]. 
\end{align*}
With~$X$ and~$U$ in Assumption~\ref{reg:ass4}, $G_2 = U - G_1 X$ is computed as
\begin{align*}
G_2 = 
\left[\begin{array}{cccc}
0.869 & -0.0019 & 0.873 & 0.174\\
-0.0019 & 0.869 & 0.174 & 0.873
\end{array}\right]. 
\end{align*}

Second, the LMIs (\ref{LMI2}) and (\ref{LMI}) have solutions $P$ and $\hat L_1$ for $\gamma = 0.365$. The matrix $L_1 = P^{-1} \hat L_1$ is 
\begin{align*}
L_1 = \left[\begin{array}{rrrrr}
-0.193 & 0.0088 & 0.0828 & 0.0111\\
0.0088 & -0.193 & 0.0111 & 0.0828\\
-0.0717 & 0.0072 & -0.134 & -0.0129\\
 0.0072 & -0.0717 & -0.0129 & -0.134\\
0.0253 & -0.0253 & -0.0504 & 0.0504
\end{array}\right].
\end{align*}
In this scenario, the initial state of the controller is chosen as~$[0 \ 0 \ 380 \ 380 \ 0]^\top$ because the state of the controller takes this value when the control objective is achieved. 

Suppose that each user adds the Gaussian noise to~$I_i$ and~$V_i$ before sending them to the power company. Based on our observation for input observability, we design input noises from the principal components of~$N_{10,5}^\top N_{10,5}$ of the controller, where the initial state of the controller is assumed to be a piece of public information. Its eigenvalues are shown in Fig.~\ref{eig:fig}. Let~$v_{j,i}$ be the projection of the normalized eigenvectors corresponding to the eigenvalue~$\lambda_j$ onto the~$u_i(0)$-space. By using non-zero~$\lambda_j$, we compute
\begin{align*}
\sum_{j=21}^{40} \lambda_j v_{1,j} v_{1,j}^\top
= \sum_{j=21}^{40} \lambda_j v_{2,j} v_{2,j}^\top = \left[ \begin{array}{rr}
0.0347 & -0.0106\\
-0.0106 & 0.0129
\end{array}\right].
\end{align*}
Since larger~$\lambda_j$ characterizes less private information of~$u_i(0)$, it is reasonable to add larger noise to~$u_i(0)$ corresponding to larger~$\lambda_j$. Therefore, we scale by a positive constant $a$, namely
\begin{align*}
\overline \Sigma_1 
=  a^2 \left[ \begin{array}{rr}
0.0347 & -0.0106\\
-0.0106 & 0.0129
\end{array}\right]
\end{align*}
as the covariance matrix of the input noise for each user. The condition~\eqref{epinoise} for $(\varepsilon,\delta)$-differential privacy holds if
\begin{align*}
a \ge 10.8 c \cR (\varepsilon,\delta). 
\end{align*}
Let~$c = 1$. In privacy related literatures in systems and control~\cite{HE:17,NP:14,LM:18}, $\varepsilon$ and~$\delta$ are chosen to be values in $[0.3, 1.6]$ and $[0.01,0.05]$, respectively. We use similar values.
For instance, for~$\varepsilon = 0.3$ and~$\delta = 0.0446$ or~$\varepsilon = 0.42$ and~$\delta = 0.00820$, the condition holds for~$a = 64.3$. For~$\varepsilon = 0.3$ and~$\delta = 0.0446$ or~$\varepsilon= 0.69$ and~$\delta = 0.00820$, the condition holds for~$a=39.7$. For~$\varepsilon=1.4$ and~$\delta=0.0446$, the condition holds for~$a = 15.8$.

\begin{figure}[!t]
\begin{center}
\includegraphics[width=80mm]{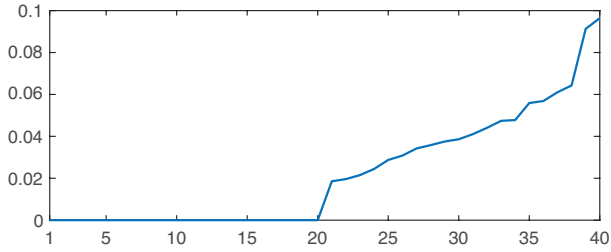}\vspace{-5mm}
\end{center}
\caption{The value of each eigenvalue of~$N_{10,5}^\top N_{10,5}$}
\label{eig:fig}
\end{figure}

Figure~\ref{sim:fig} shows~$y_p$ and~$u_p$ of the closed-loop system for the four cases: no noise,~$a=15.8$,~$a=39.7$, and~$a=64.3$. If there is no noise, the tracking error converges to zero. However, the change of $I_1$ affects clearly~$I_2$,~$V_2$, and~$u_2$. Therefore, user~$2$ can identify that user~$1$ starts to use more electricity. In contrast, privacy-preserving controllers with noises mask the effects caused by the electricity consumption of user $1$ against user $2$. Although adding large noise increases the privacy level, it unfortunately makes~$I_i$,~$V_i$, and~$u_i$ fluctuate. Therefore, the balance between required control and privacy performances is needed when designing the noise. In this case, the noise with~$a=15.8$ is enough to protect the fluctuation of the input~$u_2$.

\begin{figure}[tb]
\begin{center}
\includegraphics[width=85mm]{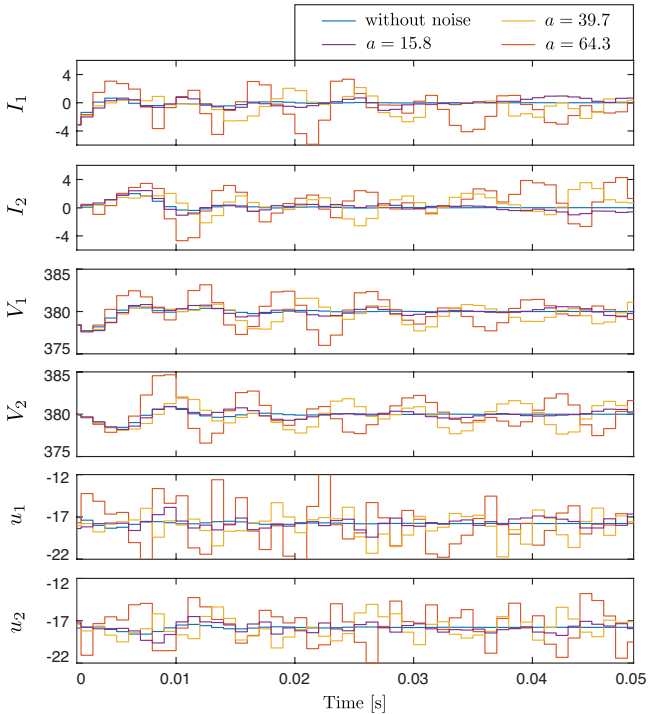}\vspace{-5mm}
\end{center}
\caption{The outputs and inputs of the closed-loop system controlled by the privacy-preserving controller}
\label{sim:fig}
\end{figure}

\section{Toward Nonlinear Mechanisms}\label{NDP:s}
The objective in this section is to extend some of our results to nonlinear mechanisms toward nonlinear privacy-preserving controller design. The output regulation and $H_{\infty}$-norm analysis are extended to nonlinear systems at least locally; see e.g.~\cite{LB:96,Huang:04}. Therefore, if differential privacy analysis is extended, one can design a nonlinear privacy-preserving controller at least locally in a similar manner as the linear case. In this section, we proceed with differential privacy analysis of the Gaussian mechanism induced by a nonlinear dynamical system and output Gaussian noise. For nonlinear dynamical systems, even if Gaussian noise is added to the input channel, the output variable is not Gaussian in general, and thus we do not analyze the mechanisms induced by input noise. 

\subsection{Differential Privacy with Output Noise}
Consider the following nonlinear discrete-time control system with output noise
\begin{align}
\left\{\begin{array}{l}
x(t+1) = f(x(t),u(t)),\\
y_w(t) = h(x(t),u(t)) + w(t),
\end{array}\right. \label{nonlinear}
\end{align}
where~$f:\bR^n \times \bR^m \to \bR^n$ and $h: \bR^n \times \bR^m \to \bR^q$ are continuous. Its solution $x(t)$ starting from $x_0$ controlled by $U_{t-1}$ is denoted by $\phi(t,x_0,U_{t-1})$, where $\phi(0,x_0,U_{-1}) :=x_0$. The output sequence $Y_{w,t}$ can be described by
\begin{align}
Y_{w,t} = H_t(x_0,U_t) + W_t, \label{disoutseq}
\end{align}
where $H_t: \bR^n \times \bR^{(t+1)m} \to \bR^{(t+1)q}$ is
\begin{align}
H_t(x_0,U_t)
:=\left[\begin{array}{c}
h(\phi(0,x_0,U_{-1}),u(0)) \\ h(\phi(1,x_0,U_0),u(1)) \\ \vdots \\ h(\phi(t,x_0,U_{t-1}),u(t)) 
\end{array}\right]. \label{Ht}
\end{align}
Now, we are ready to obtain an extension of Theorem~\ref{dpr:thm} to the nonlinear Gaussian mechanism by using input data dependent Gaussian noise.

\begin{secthm}\label{dpr_2:thm}
The Gaussian mechanism (\ref{disoutseq}) induced by $W_t \sim \cN_{(t+1)q}(\mu (x_0,U_t),\Sigma (x_0,U_t))$ is $(\varepsilon,\delta)$-differentially private for ${\rm Adj}_2^c$ at a finite time $t$ with $\varepsilon (x_0,U_t)> 0$ and $1/2 > \delta (x_0,U_t)> 0$ if the covariance matrix $\Sigma(x_0,U_t) \succ 0$ is chosen such that
\begin{align}
&\inf_{| [\bar x_0; \bar U_t]  |_2  \le c}\frac{1}{\bar H_t(x_0, U_t, \bar x_0, \bar U_t)} \ge c \cR \left(\varepsilon (x_0,U_t),\delta (x_0,U_t)\right),\label{ep_non}
\end{align}
where
\begin{align*}
&\bar H_t(x_0, U_t, \bar x_0, \bar U_t)\\
&:= | H_t(x_0+ \bar x_0,U_t+ \bar U_t) - H_t(x_0,U_t) |_{\Sigma^{-1}(x_0,U_t)}
\end{align*}
for any $(x_0 ,U_t ) \in \bR^n \times \bR^{(t+1)m}$.
\end{secthm}

\begin{proof}
In a similar manner as for Theorem~\ref{dpr:thm}, one obtains (\ref{z,delta}) for
\begin{align*}
z=|H_t(x_0, U_t) - H_t(x'_0, U'_t)|_{\Sigma^{-1}(x_0,U_t)}^{-1}.
\end{align*}
Define $\bar x_0 = x'_0 - x_0$ and $\bar U_t = U'_t - U_t$. Then, $((x_0,U_t),(x'_0,U'_t)) \in {\rm Adj}_2^c $ implies~$| [\bar x_0; \bar U_t]  |_2  \le c$. It follows that
\begin{align*}
z= \bar H_t(x_0, U_t, \bar x_0, \bar U_t) \le \sup_{| [\bar x_0; \bar U_t]  |_2  \le c}\bar H_t(x_0, U_t, \bar x_0, \bar U_t),
\end{align*}
for any $((x_0,U_t),(x'_0,U'_t)) \in {\rm Adj}_2^c$. Therefore, if (\ref{ep_non}) holds, (\ref{z,delta}) holds.
\end{proof}

In Theorem~\ref{dpr_2:thm}, the mean value and variance can be made functions of~$(x_0,U_t)$ under the reasonable assumption that a system manager designing noise knows the initial state and inputs of the system. Even if the system manager does not know them exactly, noise can still be designed by using a constant mean value and variance. Note that using $(x_0,U_t)$-dependent noise does not break privacy guarantee because information of noise is not published in general.

The paper~\cite{LeNy:20} studies nonlinear observer design based on differential privacy in the contraction framework with constant metrics. The provided results can be extended to differential privacy analysis of nonlinear systems, but as they stand, only for stable systems. In contrast, Theorem~\ref{dpr_2:thm} can be used for unstable systems and considers a more general mean value and variance depending on~$(x_0,U_t)$.

In a similar manner as Remark~\ref{Laplace:rem}, for the i.i.d. Laplace noise $w_i(t)$, $i=1,\dots,q$, $t \in \bZ_+$ with the variance $\mu (x_0,U_t) \in \bR$ and distribution $b(x_0,U_t) >0$, the mechanism (\ref{disoutseq})  is $(\varepsilon,0)$-differentially private for ${\rm Adj}_1^c$ at a finite time $t$ with $\varepsilon > 0$ if
\begin{align}
b(x_0,U_t) \ge \sup_{| [\bar x_0; \bar U_t]  |_1 \le c}\frac{\left| H_t(x_0+ \bar x_0,U_t+ \bar U_t) - H_t(x_0,U_t) \right|_1}{\varepsilon (x_0,U_t)}
\label{laplan_non}
\end{align}
for any $(x_0,U_t) \in \bR^n \times \bR^{(t+1)m}$. Furthermore, suppose that $f$ and $h$ are smooth. Let $\gamma(s)=x_0 + s (\bar x_0 - x_0)$ and $\nu(s)=U_t + s (\bar U_t - U_t)$ for $s \in [0,1]$. Then, 
\begin{align*}
&\left| H_t(x_0+ \bar x_0,U_t+ \bar U_t) - H_t(x_0,U_t) \right|_1 \nonumber\\
&= \left| \int_0^1 \frac{\partial H_t(\gamma(s),\nu(s))}{\partial (x_0,U_t)} 
\left[\begin{array}{c}
x_0 \\ U_t
\end{array}\right] ds \right|_1\\
&\le c \left| \int_0^1 \frac{\partial H_t(\gamma(s),\nu(s))}{\partial (x_0,U_t)} ds \right|_1\\
&\le c \sup_{(x_0,U_t) \in \bR^n \times \bR^{(t+1)m}}\left| \frac{\partial H_t(x_0,U_t)}{\partial (x_0,U_t)} \right|_1. 
\end{align*}
Therefore, (\ref{laplan_non}) holds if
\begin{align*}
b(x_0,U_t)
\ge \frac{c}{\varepsilon (x_0,U_t)} \sup_{(x_0,U_t) \in \bR^n \times \bR^{(t+1)m}}\left| \frac{\partial H_t(x_0,U_t)}{\partial (x_0,U_t)} \right|_1.
\end{align*}
 Note that in the linear case $\partial H_t(x_0,U_t)/\partial (x_0,U_t)$ is nothing but the matrix $[\begin{array}{cc}
O_t & N_t
\end{array}]$.

For Laplacian noise, differential privacy is characterized by the matrix $\partial H_t(x_0,U_t)/\partial (x_0,U_t)$. This matrix has a strong connection with the \emph{local strong input observability} of the nonlinear system~(\ref{nonlinear}); the concept of strong observability can be extended to nonlinear systems as for local observability~\cite{NS:90} based on the distinguishability of a pair of initial states and initial inputs. In fact, one can derive a necessary and sufficient condition for \emph{local strong input observability} in terms of the differential one-forms corresponding to $\partial H_t(x_0,U_t)/\partial (x_0,U_t)$ as follows: there exists $t \in \bZ_+$ such that
\begin{align*}
{\rm span}\{{\rm d}H_t(x_0,U_t)\} \cap {\rm span}\{{\rm d}x, {\rm d}u_0\} = {\rm span}\{{\rm d}x, {\rm d}u_0\}
\end{align*}
under the constant dimensional assumption for all $(x_0,U_t) \in \bR^n \times \bR^{(t+1)q}$; see e.g.~\cite{NS:90} for similar discussions for local observability. This is an extension of the condition (\ref{InpinvK}). In contrast to the qualitative criterion for strong input observability, it is still not straightforward to extend the concept of Gramians. In fact, there is no clear extension of Gramians to nonlinear systems even for controllability and observability although the concept of controllability and observability and their corresponding energy functions have been extended~\cite{NS:90,Scherpen:93,KS:17}.

\subsection{Incrementally Input-to-Output Stable Systems}\label{IIOSS:s}
In Section~\ref{IO:s}, we mention that the $H_{\infty}$-norm gives an upper bound of the differential privacy level. This observation can help the privacy-preserving controller design. In this subsection, we aim at extending this result to the nonlinear case based on the concept of the incremental input-to-output stability (IOS).

For a nonlinear system, several types of gains (or called estimations) are defined; e.g. see~\cite{Sontag:08}. Especially, $L^2 \to L^2$ estimation is extended to nonlinear systems as input-to-state stability (ISS)~\cite{Sontag:08}, which is also extended to incremental properties in~\cite{Angeli:09}. Incremental ISS can be readily extended to input-to-output operators, discrete-time systems, and arbitrary $L^p \to L^p$ estimations as follows. In Appendix, we give its Lyapunov characterization.
\begin{secdefn}\label{IOS:def}
A nonlinear system~(\ref{nonlinear}) is said to be \emph{incrementally IOS} (with respect to the $p$-norm) if the output $h(\phi(t,x_0,U_{t-1}),u(t))$ exists for all $t\in\bZ_+$, for any $x_0 \in \bR^n$ and $u:\bZ_+ \to \bR^m$, and there exist class $\cK$ functions $\alpha$ and $\gamma$  such that
\begin{align}
&\sum_{\tau = 0}^t  | h(\phi(\tau,x_0,U_{t-1}),u(t)) - h(\phi(\tau,x'_0,U'_{t-1}),u'(t)) |_p \nonumber\\
&\le \alpha (|x_0 - x'_0 |_p) + \sum_{\tau = 0}^t\gamma (|u(\tau) - u'(\tau) |_p), \ t \in \bZ_+ \label{iIOS}
\end{align}
for any $(x_0,x'_0) \in \bR^n \times \bR^n$ and $(U_t,U'_t) \in \bR^{(t+1)m} \times \bR^{(t+1)m}$.
\red
\end{secdefn} 

In fact, $\alpha$ and $\gamma$ do not need to belong to class $\cK$ for differential privacy analysis, and non-negative functions are enough. To connect differential privacy analysis with ISS, we consider class $\cK$ functions. 

In the linear case, as shown in Corollary~\ref{dpr_sta:thm}, the $H_{\infty}$-norm can be used for designing the Gaussian noise. Now, we obtain an extension of Corollary~\ref{dpr_sta:thm} to the nonlinear IOS system based on Theorem~\ref{dpr_2:thm}. The proof directly follows, and thus is omitted.
\begin{seccor}\label{IOS:cor}
Let $W_t \sim \cN_{(t+1)q}(\mu (x_0,U_t),\Sigma (x_0,U_t))$ be a non-degenerate multivariate Gaussian noise. Then, the Gaussian mechanism (\ref{disoutseq}) induced by an incrementally IOS nonlinear system~(\ref{nonlinear}) (with respect to $2$-norm) is $(\varepsilon,\delta)$-differentially private for  ${\rm Adj}_2^c$ at a finite time $t$ with $\varepsilon (x_0,U_t)> 0$ and $1/2 > \delta (x_0,U_t)> 0$ if the covariance matrix $\Sigma(x_0,U_t) \succ 0$ is chosen such that
\begin{align*}
&\lambda_{\min}^{1/2}(\Sigma (x_0,U_t))\nonumber\\
&\ge (\alpha(c) +(t+1) \gamma (c)) \cR \left(\varepsilon (x_0,U_t),\delta (x_0,U_t)\right)
\end{align*}
for any $x_0 \in \bR^n$ and $U_t \in \bR^{(t+1)m}$. 
\red
\end{seccor}

\section{Conclusion}~\label{Con:s}
In this paper, we have studied differential privacy of Gaussian mechanisms induced by discrete-time linear systems. First, we have analyzed differential privacy in terms of strong input observability and then have clarified that the differential privacy level is characterized by the maximum eigenvalue of the input observability Gramian. In other words, small noise is enough to make the less input observable Gaussian mechanism highly differentially private. Moreover, we have shown that the mechanisms induced by input and output noises have the same differential privacy level for suitable covariance matrices. Next, we have developed a privacy-preserving controller design method, which can make a linear system highly private by adding small noise. Finally, we have briefly mentioned differential privacy analysis of incrementally IOS nonlinear systems. 

Although we have focused on differential privacy in this paper, our analysis and controller design can be tools for studying more general privacy issues of control systems. Differential privacy has been originally proposed in static data analysis, and one may extend this concept to dynamical systems further or develop new privacy concepts for dynamical systems. For privacy-preserving controller design, we have first designed a controller satisfying a certain control performance and then added noise to protect private information. There remain several interesting research directions. One is to investigate an updating method for the covariance matrix of noise forward in time based on the idea of Kalman filter design. Another is to develop a randomized control mechanism guaranteeing a certain control performance, which may enable us to design the controller and noise at the same time.

In general, measurement and input noises, disturbance and model error make analysis and controller design difficult and deteriorate the control performance, and thus they are regarded as troubles. However, they improve the system's privacy level. Therefore, the privacy-preserving controller design reduces to the trade-off between the privacy level and control performance.

{\it Acknowledgement:} 
We thank Dr. Michele Cucuzzella, University of Groningen for fruitful discussions on DC microgrids.

\appendix
\renewcommand{\thesecthm}{\Alph{secthm}}
\section{Incremental IOS Analysis}
In this appendix, we provide a sufficient condition for incremental IOS. 
\begin{secthm}\label{IOS:thm}
A nonlinear system~(\ref{nonlinear}) is incrementally IOS if there exist a continuous function $V:\bR^n \times \bR^n \to \bR_+$, constants $c_1>0$, $\lambda \in (0,1)$, class $\cK$ functions $\sigma_1,\sigma_2$, and a class $\cK_{\infty}$ function $\alpha_2$ such that
\begin{align}
&\hspace{-2mm}c_1|h(x_0,u)-h(x'_0,v) |_p \le  V(x_0,x'_0)  + \sigma_1(|u - u'|_p),\label{firstcond}\\
&\hspace{-2mm}V(x_0,x'_0)\le \alpha_2(|x_0 - x'_0|_p), \label{firstcond2}\\
&\hspace{-2mm}V(f(x_0,u), f(x'_0,u') ) \le  \lambda V(x_0,x'_0)  + \sigma_2 (|u - u'|_p)\label{second}
\end{align}
for any $(x_0,x'_0) \in \bR^n\times \bR^n$ and $(u,u') \in \bR^m\times \bR^m$.
\end{secthm}
\begin{proof}
Recursively using the inequality (\ref{second}) for $\tau \ge 1$ yields
\begin{align*}
&V(\phi(\tau,x_0,U_{\tau -1}), \phi(\tau,x'_0,U'_{\tau -1}))\\
&\le \lambda  V(\phi(\tau-1,x_0,U_{\tau -2}), \phi(\tau-1,x'_0,U'_{\tau -2})) \\
&\hspace{5mm} +  \sigma (|u(\tau-1)-u'(\tau-1)|_p)\\
&\le \lambda^2  V(\phi(\tau-2,x_0,U_{\tau -3}), \phi(\tau-2,x'_0,U'_{\tau -3})) \\
&\hspace{5mm}+ \lambda  \sigma_2 (|u(\tau-2)-u'(\tau-2)|_p ) \\
&\hspace{5mm} +  \sigma_2 (|u(\tau-1)-u'(\tau-1)|_p)\\
&\le \lambda^\tau   V(x_0,x'_0) + \sum_{r = 0}^{\tau-1} \lambda^{\tau-1-r} \sigma_2 (|u(r)-u'(r)|_p).
\end{align*}
From (\ref{firstcond}),
\begin{align*}
&c_1|h(\phi(\tau,x_0,U_{\tau-1},u(\tau)) - h(\phi(\tau,x'_0,U'_{\tau-1}),u'(\tau))|_p \\
&\le \lambda^\tau  \alpha_2 (| x_0 - x'_0 |_p) + \sum_{r = 0}^{\tau-1} \lambda^{\tau-1-r} \sigma_2 (|u(r)-u'(r)|_p) \\
&\hspace{5mm}+ \sigma_1(|u(\tau) - u'(\tau)|_p).
\end{align*}
By taking the summation, we have
\begin{align*}
&c_1 \sum_{\tau =0}^t |h(\phi(\tau,x_0,U_{\tau-1},u(\tau)) - h(\phi(\tau,x'_0,U'_{\tau-1}),u'(\tau))|_p\\
&\le \sum_{\tau =0}^t\Biggl( \lambda^\tau  \alpha_2 (| x_0 - x'_0 |_p) + \sigma_1(|u(\tau) - u'(\tau)|_p) \\
&\hspace{5mm}+  \sum_{r = 0}^{\tau-1} \lambda^{\tau-r-1} \sigma_2 (|u(r)-u'(r)|_p)\Biggr)\\
&\le\frac{1-\lambda^t}{1-\lambda}  \alpha_2 (| x_0 - x'_0|_p)+ \sum_{\tau =0}^t \sigma_1(|u(\tau) - u'(\tau)|_p)\\
&\hspace{5mm}+ \sum_{r=0}^{t-1} \frac{1-\lambda^{t-1-r}}{1-\lambda} \sigma_2 (|u(r)-u'(r)|_p)\\
&\le\frac{\alpha_2 (| x_0 - x'_0 |_p)}{1-\lambda}\\
&\hspace{5mm} + \sum_{r=0}^t \Biggl(\frac{\sigma_2 (|u(r)-u'(r)|_p)}{1-\lambda} +\sigma_1(|u(\tau) - u'(\tau)|_p)\Biggr),
\end{align*}
where in the second last inequality, $\lambda \in (0,1)$ is used. Therefore, the system is incrementally IOS. 
\end{proof}
\begin{secrem}
We mentioned that for differential privacy analysis, $\alpha$ and $\gamma$ are required to be only non-negative functions. Here, we obtain a similar characterization by using non-negative functions $\sigma_1$, $\sigma_2$, and $\alpha_2$.
\red
\end{secrem}

\bibliographystyle{IEEEtran} 
\bibliography{DPCRef}
\end{document}